\documentclass[journal,twoside,web]{ieeecolor}
\usepackage{generic}
\usepackage{cite}
\usepackage{amsmath,amssymb,amsfonts}
\usepackage{algorithmic}
\usepackage{graphicx}
\usepackage{algorithm,algorithmic}
\usepackage{hyperref}
\hypersetup{hidelinks=true}
\usepackage{textcomp}

\usepackage{subcaption}
\usepackage{cancel}
\usepackage{changes}

\usepackage{todonotes}
\usepackage{mathrsfs}

\definecolor{deepblue}{rgb}{0, 0.2, 0.6}

\usetikzlibrary{positioning}
\usetikzlibrary{shapes,arrows}
\newtheorem{remark}{\bf Remark}

\newtheorem{lemma}{\bf Lemma}
\newtheorem{definition}{\bf Definition}
\newtheorem{problem}{\bf Problem}

\newtheorem{thm}{\bf Theorem}

\usepackage{xcolor}  %
\definecolor{customgreen}{RGB}{0,102,51}  
\makeatletter
\renewcommand{\subsection}{%
  \@startsection{subsection}{2}{\z@}%
    {1.5ex plus 1.5ex minus .5ex}%
    {0.7ex plus .5ex minus 0ex}%
    {\normalfont\normalsize\itshape\color{customgreen}}%
}
\makeatother

\makeatletter
\renewcommand{\subsubsection}{%
  \@startsection{subsubsection}{3}{\z@}%
    {0pt}%
    {0pt}%
    {\normalfont\normalsize\bfseries\itshape\color{deepblue}}%
}
\makeatother

\newcommand{\wxy}{\color{black}}
\newcommand{\hw}{\color{black}}
\newcommand{\hwz}{\color{black}}
\newcommand{\wxyr}{\color{black}}

\usepackage{caption}
\usepackage{todonotes}
\usepackage{siunitx}
\allowdisplaybreaks

\def\BibTeX{{\rm B\kern-.05em{\sc i\kern-.025em b}\kern-.08em
    T\kern-.1667em\lower.7ex\hbox{E}\kern-.125emX}}
\begin{document}
% \title{A class of unified disturbance rejection control barrier functions}
\title{{\wxyr Disturbance Rejection Control Barrier Functions}}
\author{Xinyang Wang, Wei Xiao, and Hongwei Zhang, \IEEEmembership{Senior Member,~IEEE}
\thanks{This work was supported by the National Natural Science Foundation of China under grants 62473114 and W2541002. (Corresponding author: Hongwei Zhang)}
\thanks{Xinyang Wang and Hongwei Zhang are with the School of Intelligence Science and Engineering, Harbin Institute of Technology, Shenzhen, Guangdong 518055, China (e-mail: wangxy@stu.hit.edu.cn; hwzhang@hit.edu.cn).}
\thanks{Wei Xiao is with Massachusetts
Institute of Technology, Cambridge, MA 02139, USA (e-mail: weixy@mit.edu).}
}

\maketitle

\begin{abstract}
%
% Robust control barrier functions (CBFs) typically require a prior knowledge of disturbance bounds, posing severe challenges in real-world applications. 
% %
% Moreover, most robust CBFs are limited to matched disturbances, where control inputs can directly compensate for disturbances.
Most existing robust control barrier functions (CBFs) can only handle matched disturbances, restricting their applications in real-world scenarios. 
While some recent advances extend robust CBFs to unmatched disturbances, they heavily rely on differentiability property of disturbances, and fail to accommodate non-differentiable case for {\wxyr safety constraints with high relative degree}.
% %
% Such requirements fail to capture real-world disturbances like wind gusts or mechanical vibrations that often exhibit bounded but non-differentiable characteristics.
%
To address these limitations, this paper proposes {\wxyr a class of disturbance rejection CBFs (DRCBFs), including knowledge-based DRCBFs (kDRCBFs) and reciprocal-compensated DRCBFs (rDRCBFs).}
These two DRCBFs can strictly guarantee safety under general bounded disturbances, which includes both matched or unmatched, differentiable or non-differentiable disturbances as special cases.
%
% This class of DRCBFs guarantees safety in the presence of bounded disturbances with arbitrary form, which is highly favorable for enhancing robustness of CBF strategies.
%
{\wxyr Moreover, no information of disturbance is needed in rDRCBFs.}
Simulation results illustrate that the proposed  DRCBFs outperform existing robust CBFs. 
\end{abstract}


%%======================
\section{Introduction}
%%======================
%
Safety concern is a central issue for controller design in safety-critical scenarios, such as autonomous driving \cite{Koopman2017Autonomous} and human-robot collaboration \cite{Li2024Optimal}. 
Inspired by control Lyapunov function (CLF), the concept of control barrier function (CBF) was introduced in \cite{ames2017control} to transform state-dependent safety constraints into a control-affine formulation.
By incorporating CBF and CLF into a quadratic program, a CLF-CBF {\wxyr quadratic-program}-based controller was designed for safety guarantee \cite{ames2019control}.
The effectiveness of CBF strategies heavily relies on  accurate system models, which are, however, generally unknown in practice due to unknown disturbances/uncertainties.
This gap may render CBF strategies ineffective, potentially leading to unsafe or even dangerous behavior of the system.

A natural idea to handle disturbances is to design robust CBFs that account for worst-case scenarios \cite{jankovic2018robust, nguyen2021robust, buch2022robust}.
Although this approach provides strict safety guarantees,
{\wxy
it is often overly conservative in the sense that the state of the system is kept far away from the boundary of the safe set \cite{wang2023disturbance}}.
Moreover, large disturbance bounds may render the {\wxyr quadratic program} infeasible\cite{xiao2022adaptive}.
{\wxyr Inspired by input-to-state stability in control theory}, input-to-state safety CBF was developed in \cite{kolathaya2019input,alan2022safe,alan2023control} {\wxyr to relax the strict safety guarantee, allowing} the system trajectory to enter a slightly enlarged safe set.
However, these above-mentioned approaches (\cite{jankovic2018robust, nguyen2021robust, buch2022robust}, \cite{kolathaya2019input,alan2022safe,alan2023control}) are {\wxyr limited} to matched disturbances, i.e., disturbances coming into CBF condition through the same channel as control inputs.
This significantly restricts their practical applications, where disturbances may be unmatched.
Moreover, these approaches can only handle constraints with relative degrees of one or two, and cannot {\wxyr handle arbitrary-relative-degree} constraints, which is more general, yet much harder. 
Therefore, constructing robust CBFs for constraints with {\wxyr arbitrary relative degree under unmatched disturbances} is worthy of further investigation.

High-order CBF {\wxyr (HOCBF)} was proposed in \cite{xiao2021high} to handle constraints with arbitrary relative degree.
Unfortunately, unmatched disturbances may introduce their derivatives of different orders in {\wxyr HOCBF} design \cite{xiao2024robust}, which further complicates the safe control problem.
To estimate unmatched disturbances and their derivatives of different orders, disturbance observers (DOs) are {\wxyr recently} integrated with high-order CBFs to actively reject disturbances  \cite{ersin2025safety, wang2025safety}.
{\hw But DO-based CBFs strongly rely on the assumption that the disturbances {\wxyr are} differentiable to a certain order, and the bounds of both disturbances and their derivatives of different orders {\wxyr are} known. This is, however, not the case in most practical applications.} For example, stochastic disturbances such as wind gusts \cite{qian2020Path} are typically non-differentiable.
Non-differentiable unmatched disturbances pose significant challenges for {\wxyr HOCBF} design.
To address this issue, \cite{breeden2023robust} transformed the original high-relative-degree constraint into a robust {\wxyr form with} relative degree one, thus avoiding the derivatives of disturbance.
{\hw However, constructing such CBFs is very tricky, if possible, when constraint has a relative degree greater than two.
Consequently, constructing a unified robust CBF capable of handling constraints with arbitrary relative degree under general bounded disturbances, either differentiable or non-differentiable, either matched or unmatched, is very interesting, {\hwz yet remains largely open}.}

It is also worth noting that almost all existing robust CBF designs, such as worst-case CBFs \cite{jankovic2018robust, nguyen2021robust, buch2022robust}, \cite{xiao2022adaptive}, \cite{breeden2023robust}, input-to-state safety CBFs \cite{kolathaya2019input, alan2022safe, alan2023control}, and DO-based CBFs \cite{ersin2025safety,wang2025safety},  \cite{immersion2024wang,zhou2025temporal,Peng2024Adaptive} require the knowledge of disturbance bound.
While a recent work \cite{sun2024safety} {\wxyr eliminates} the need for disturbance bounds, it poses a {\wxyr strong} assumption that disturbances must be generated by an exosystem with known dynamics.
Both requirements significantly restrict the practical applications of robust CBFs, as disturbances in real-world systems (e.g., white noise) often lack known bounds and structured dynamics.
{\wxyr
To address this limitation, a robust control barrier function was proposed in a very recent work \cite{shen2025composite} to reject disturbances without any a priori knowledge.
However, this approach cannot handle high-relative-degree constraints under unmatched disturbances.
}

{\hw Motivated by the aforementioned statements, this paper aims to construct {\wxyr a new class of disturbance rejection CBFs (DRCBFs)} capable of rejecting general bounded disturbances.}
%
% the key limitations of existing literature in robust CBFs are summarized as follows:
% 1) they cannot reject non-differentiable unmatched bounded disturbances;
% 2) either disturbance dynamics or bounds are needed.
%
% In this paper, we address these limitations by studying a general robust control problem, where disturbance is unmatched (with matched disturbance as a special case), bounded with unknown structure, and unknown disturbance bound (known bounds as a special case).
% %
% Then we propose a class of disturbance rejection CBFs (DRCBFs) for strict safety guarantee of system subject to general bounded disturbances.
First, we assume the disturbance bound is known and propose a {\wxyr knowledge-based DRCBF (kDRCBF)} by recursively differentiating CBFs and defining its {\wxyr conservative} form  using disturbance bounds. Then we propose {\wxyr a reciprocal-compensated DRCBF (rDRCBF) by replacing the disturbance bound with a reciprocal-like term}, whose magnitude grows to infinity as the {\wxyr state trajectory} approaches the boundary of the safe set, thus {\wxyr removing the requirement of any knowledge of disturbances.}
% these limitations by studying a general robust control problem, where disturbance is unmatched (with matched disturbance as a special case), bounded with unknown structure, and unknown disturbance bound (known bounds as a special case).
% %
% Then we propose a class of disturbance rejection CBFs (DRCBFs) for strict safety guarantee of system subject to general bounded disturbances.
%
The main contributions are summarized as follows:

{\wxy
\begin{enumerate}
\item
{\wxyr
A knowledge-based DRCBF is proposed to  reject general bounded disturbances.
Compared with existing robust CBFs for unmatched disturbances \cite{ersin2025safety}, \cite{wang2025safety}, \cite{breeden2023robust}, kDRCBF can handle constraints with arbitrary relative degree under non-differentiable unmatched disturbances.}
% eliminates restrictive assumptions on the continuity of disturbance and the relative degree of system while guaranteeing safety.

\item 
{\wxyr  
On the basis of kDRCBF, we further propose a reciprocal-compensated DRCBF that strictly guarantees safety without needing any knowledge of the disturbances. 
This capability distinguishes our rDRCBF from} existing robust CBFs which require either disturbance bounds \cite{jankovic2018robust, nguyen2021robust,buch2022robust}, \cite{kolathaya2019input,alan2022safe,alan2023control,xiao2024robust,ersin2025safety,wang2025safety}, \cite{immersion2024wang,zhou2025temporal,Peng2024Adaptive}, or the disturbances dynamics \cite{sun2024safety}.
%
% To the best of our knowledge, this is the first robust CBF that does not rely on any prior knowledge of disturbances.
{\wxyr Furthermore, compared with \cite{shen2025composite}, our rDRCBF  can handle constraints with arbitrary relative degree under unmatched disturbances.}

\item
We develop a systematic parameter design for DRCBFs to reduce conservativeness.
Further, we introduce tunable parameters into the {\wxyr rDRCBFs} framework, enabling flexible expansion of the safe set.
This allows the state to get closer to the boundary of the original safe set without sacrificing safety guarantees, thus effectively balancing robustness and conservativeness.
%
% {\wxy 
% Reciprocal CBF is an energy function that grows infinity when state approaches the boundary of safe set.
% %
% Utilizing this property, our adaptive DRCBF introduces reciprocal CBF into high-order CBF to over-approximate disturbance term, thus rendering a subset of the safe set forward invariant.
% %
% This safety subset can be flexibly enlarged via tunable parameters, which exhibits non-conservative characteristic of our adaptive DRCBF while maintaining robust safety.
% }
% {\color{red}what are the advantages of ZCBFs and RCBFs?}
% %
% reduce conservativeness while maintaining robust safety.
\end{enumerate}
}

\emph{\textbf{Notations:}} 
The set of real numbers, positive real numbers, non-negative real numbers and nonnegative integers are denoted by $\mathbb{R}$, $\mathbb{R}_{> 0}$, $\mathbb{R}_{\geq 0}$ and $\mathbb{N}$.
Given any $i,j \in \mathbb{N}$ and $i < j$, define $\mathbb{N}_{i:j} = \{ i,i+1,\cdots,j \}$.
Both Euclidean norm of a vector and Frobenius norm of a matrix are denoted by $\lVert \cdot \rVert$.
{\wxyr
A continuous function $\alpha:\mathbb{R}_{\geq 0}\rightarrow\mathbb{R}_{\geq0}$ is said to belong to class $\mathcal{K}_\infty$ if $\alpha(0) = 0$, $\lim_{r\rightarrow\infty}\alpha(r) = \infty$ and it is strictly increasing.
A continuous function $\alpha:\mathbb{R}\rightarrow\mathbb{R}$ is said to belong to extended class $\mathcal{K}_\infty$ if $\alpha(0) = 0$, $\lim_{r\rightarrow-\infty}\alpha(r) = -\infty$, $\lim_{r\rightarrow\infty}\alpha(r) = \infty$ and it is strictly increasing.
}
%
% For a time-varying bounded signal $v(t)$, $\lVert v \rVert_\infty = \sup_{t \in \mathbb{R}_{\geq 0}} \lVert v(t) \rVert$.

%==================================
\section{Preliminaries and problem formulation} \label{Sec Formulation}
\subsection{High-order Control Barrier Function}
Consider a nominal nonlinear system 
\begin{align} \label{eq:nominal_sys}
    \dot{x} = f(x)+g(x)u,
\end{align}
where $x \in \mathbb{R}^n$ and $u\in \mathbb{R}^p$ are the state  and the control input of the system, respectively; $f: \mathbb{R}^n \rightarrow \mathbb{R}^n$ and   $g: \mathbb{R}^n \rightarrow \mathbb{R}^{n\times p}$ are locally Lipschitz continuous.
The safe set is described as 
\begin{align} \label{eq:safety_set}
    \mathscr{C} = \{x \in \mathbb{R}^n: \; b(x) \geq 0 \}, 
\end{align}
where $b : \mathbb{R}^n  \rightarrow \mathbb{R}$ is {\wxy sufficiently} differentiable with respect to \eqref{eq:nominal_sys}.
The input relative degree (IRD) of $b(x)$ for \eqref{eq:nominal_sys} is:
\begin{definition}\label{def:IRD} 
(IRD \cite{khalil2002nonlinear}):
% Consider a safe set $\mathscr{C}$.
%
% The input relative degree of $b(x)$ on $\mathscr{C}$ with respect to system \eqref{eq:nominal_sys} is defined as $R_I \in \mathbb{N}_{1:m}$ if $\lVert L_gL_f^{k} b(x) \rVert =0, ~\forall k\in \mathbb{N}_{0:R_I-2}$ and $\lVert L_g\mathcal{L}_f^{R_I-1} b(x) \rVert \neq 0 \text{ for all } x \in \mathscr{C}$, where $L_f b$ and $L_g b$ are Lie derivatives of $b$ along $f$ and $g$, respectively.
{\wxyr 
The input relative degree of $b(x)$ with respect to system \eqref{eq:nominal_sys} is the number of times $b(x)$ must be differentiated along \eqref{eq:nominal_sys} until the control input $u$ shows.}
% The input relative degree of $b(x)$ on $\mathscr{C}$ with respect to system \eqref{eq:nominal_sys} is defined as an integer $m$ if $L_gL_f^{k} b(x) =0, ~\forall k\in \mathbb{N}_{0:m-2}$ and $L_gL_f^{m-1} b(x)  \neq 0 \text{ for all } x \in \mathscr{C}$, where $L_f b$ and $L_g b$ are Lie derivatives of $b$ along $f$ and $g$, respectively.
\end{definition}

For sufficiently differentiable function $b(x)$ with IRD $m$, define a series of sets for all $i\in\mathbb{N}_{1:m}$ as
% {\color{red}this should be $\mathbb{N}_{1:R_I}$? The same applies to all the equations/notations below} {\wxy I have redefined the IRD as m and DRD as r.}:
\begin{align} \label{eq:classical_HOCBF_set}
    \mathscr{X}_i = \{ x \in \mathbb{R}^n :~\vartheta_{i-1}(x)\geq 0\},   
\end{align}
where $\vartheta_i$ satisfies
\begin{align} \label{eq:hocbf}
    \vartheta_i(x) = \dot{\vartheta}_{i-1}(x)+\alpha_i(\vartheta_{i-1}(x)),~\vartheta_0(x) = b(x)
\end{align}
{\wxyr and $\alpha_{i}$ is a sufficiently differentiable extended class $\mathcal{K}_\infty$ function.}
Then the high-order control barrier function (HOCBF) can defined as follows.
\begin{definition}\label{def:HOCBF} 
(HOCBF \cite{xiao2021high}, \cite{xiao2024robust})
{\wxyr Let the sets $\mathscr{X}_i$ and the functions $\vartheta_i$, $i\in\mathbb{N}_{1:m}$ be defined by \eqref{eq:classical_HOCBF_set} and \eqref{eq:hocbf}, respectively.}
A sufficiently differentiable function $b(x)$ is {\wxyr a high-order control barrier function} of {\wxyr order $m$} for system \eqref{eq:nominal_sys} {\wxyr if there exist extended class $\mathcal{K}_\infty$ functions $\alpha_1,\alpha_2,\cdots,\alpha_m$} such that $x(0) \in \mathscr{X}:=\cap_{i=1}^m \mathscr{X}_i$ and
% \begin{align} \label{eq:HOCBF}
%     \mathop{\rm{sup}}\limits_{u \in \mathbb{R}^p} \left\{ L^{m}_{f} b+ L_gL^{m-1}_{f}b~u + \mathcal{O}(x) \right\} \geq -\alpha_m(\vartheta_{m-1}),
% \end{align}
% where $\mathcal{O}(x) = \sum_{i=1}^{m-1} L_f^i(\alpha_{m-i} \circ \vartheta_{m-i-1})(x)$ and $\circ$ denotes the composition of functions.
% \end{definition}
\begin{align} \label{eq:HOCBF}
    \mathop{\rm{sup}}\limits_{u \in \mathbb{R}^p} \left\{ L_{f} \vartheta_{m-1}(x)+ L_g\vartheta_{m-1}(x)u \right\} \geq -\alpha_m(\vartheta_{m-1}(x))
\end{align}
for all $x \in \mathscr{X}$, {\wxyr where $L_f$ and $L_g$ denote the Lie derivatives along $f$ and $g$, respectively.}
\end{definition}

The existence of such $b(x)$ guarantees the existence of a control input $u$ that renders the set $\mathscr{X}$ forward invariant.
{\wxyr
In \eqref{eq:hocbf}, the extended class $\mathcal{K}_\infty$ functions are typically chosen as linear functions.
Specifically, $\alpha_i(\vartheta_{i-1})$ takes the form $p_i \vartheta_{i-1}$, where $p_i \in \mathbb{R}_{>0}$ for all $i \in \mathbb{N}_{1:m}$.}
Then \eqref{eq:hocbf} can be rewritten as a linear combination of $b^{(i)}$ (see \cite{nguyen2016exponential}), i.e.,
\begin{align} \label{eq:hocbf_linear}
    \vartheta_i(x) = b^{(i)}(x) + \sum_{j=0}^{i-1} c_{j}^i b^{(j)}(x),~i\in\mathbb{N}_{1:m},
\end{align}
where $c_j^i,~j\in \mathbb{N}_{0:i-1}$ are parameters of polynomial $\chi_i(s)=s^i + c^i_{i-1} s^{i-1} + \cdots + c^i_0$ with eigenvalues of $ p_1,p_2,\cdots,p_i$.
% {\wxy
% In literature, CBFs are commonly combined with CLFs to reformulate a safety-critial control problem into a sequence of quadratic programs (QPs), enabling computationally efficient solutions.
% }

\subsection{Problem Formulation}
In practice, system \eqref{eq:nominal_sys} is usually susceptible to external disturbances, i.e.,
\begin{align} \label{eq:dis_sys}
    \dot{x} = f\left(x\right) + g(x)u + h(x)d,
\end{align}
where $h: \mathbb{R}^n \rightarrow \mathbb{R}^{n\times q}$ is locally Lipschitz continuous, and $d \in \mathbb{R}^q$ is an unknown disturbance satisfying
$$\lVert d(t)\rVert \leq \mathcal{D},~~~\forall t\geq 0$$
for some $\mathcal{D}\in \mathbb{R}_{>0}$.
%
% {\wxyr It is assumed that disturbances do not affect the control channel .}
%

\begin{remark}
    The disturbance $d$ considered in this paper  is rather general in the sense that it can be either differentiable or non-differentiable, either matched or unmatched, either with a known bound or unknown bound, either generated randomly or by an exosystem. Thus it includes almost all bounded disturbances as special cases, such as  \cite{jankovic2018robust, nguyen2021robust,buch2022robust,wang2023disturbance,kolathaya2019input,alan2022safe,alan2023control,xiao2024robust,ersin2025safety,wang2025safety,immersion2024wang,zhou2025temporal,Peng2024Adaptive,sun2024safety}.
    Moreover, $h(x)d$ may also be unbounded if $x$ is not restricted to a compact set. 
    %
    % {\wxyr The assumption $[g,h] = 0$ is made for precluding the case where $u$ appears in $b^{(i)}$ for some $i < m$ due to the existence of $d$.}
   %  it includes the following disturbances, frequently encountered in literature, as special cases: 
   % \begin{itemize}
   %  \item Disturbance $d$ is matched and bounded by a known positive real number (cf. \cite{jankovic2018robust}).
   %  \item Disturbance $d$ is continuously differentiable and matched, and both $d$ and its derivative are bounded by known positive real numbers (cf. \cite{immersion2024wang}).
   %  \item Disturbance $d$ is continuously differentiable and unmatched, and both $d$ and its derivatives are bounded by known positive real numbers (cf. \cite{wang2023disturbance}, \cite{ersin2025safety}).
   %  \item Disturbance $d$ is unmatched and bounded by a known positive real number (cf. \cite{breeden2023robust}, \cite{xiao2024robust}).
   %  \item Disturbance $d$ is $(m-r)^{th}$ continuously differentiable and unmatched, and $d^{(i)}, \forall i \in \mathbb{N}_{0,m-r}$ are bounded by a set of known positive real numbers (cf. \cite{wang2025safety}).
\end{remark}

To quantify the effect of $d$ on {\wxyr safety constraint}, we introduce the concept of disturbance relative degree (DRD) {\wxyr of $b$ for \eqref{eq:dis_sys}.}
\begin{definition} \label{def:DRD}
(DRD):
{\wxyr 
The disturbance relative degree of $b(x)$ with respect to system \eqref{eq:dis_sys} is the number of times $b(x)$ must be differentiated along \eqref{eq:dis_sys} until the disturbance $d$ shows.}
% Consider a safe set $\mathscr{C}$ {\color{black}defined as in (\ref{eq:safety_set}) and the corresponding} sufficiently differentiable function $b(x)$ of IRD $m$ with respect to \eqref{eq:dis_sys}.
% %
% The disturbance relative degree of $b(x)$ on $\mathscr{C}$ with respect to system \eqref{eq:dis_sys} is defined as an integer $r$ if $ L_hL_f^{k} b(x) =0, ~\forall k \in \mathbb{N}_{0:r-2}$ and $ L_hL_f^{r-1} b(x) \neq 0 ~~ \textnormal{for all } x \in \mathscr{C}$, where $L_h b$ is Lie derivative of $b$ along $h$.
% \footnote{Definition \ref{def:DRD} regards the case where $r \leq m$ and the disturbance cannot change the IRD of $b$ for the nominal system \eqref{eq:nominal_sys}.}
\end{definition}

We now formulate the {\wxyr robust safe optimal control problem.}
\begin{problem} \label{problem:p1}
    Consider the disturbed system \eqref{eq:dis_sys} and the safe set $\mathscr{C}$ in \eqref{eq:safety_set}.
    Let $b(x)$ be a sufficiently differentiable function of IRD $m$ and DRD $r$ with respect to \eqref{eq:dis_sys}, where $r\leq m$. 
    Design a control law $u$ that stabilizes system \eqref{eq:dis_sys} while solving the following optimal control problem
    \begin{align}
        &\min_{u\in\mathbb{R}^q} ~~~ J\big(u(t)\big) \nonumber\\
        &~\textnormal{s.t.}~~~~~ {\wxyr x(t) \in \mathcal{X},} ~~~~ \forall t\geq 0,
    \end{align}
    where $J:\mathbb{R}^p \rightarrow \mathbb{R}$ is a cost function {\wxyr and $\mathcal{X}$ is a subset of $\mathscr{C}$.}
\end{problem}

\begin{remark}
Problem \ref{problem:p1} covers two fundamental cases when safety constraint may be violated due to disturbances, i.e., matched cases with $r = m$ when $d$ affects $b^{(m)}$ and unmatched cases with $r < m$ when $d$ affects $b^{(i)}$ for all $i\in\mathbb{N}_{r:m}$.
For the trivial case when $r > m$, $d$ will not affect  $b^{(i)}, \forall i\in\mathbb{N}_{1:m}$. This can be directly addressed by conventional HOCBF, and thus is not considered in this paper.
Since $r$ is generally unknown a priori in practice, we assume $r=1$ in the sequel for the convenience of theoretical analysis, while bearing in mind that the proposed framework works for any positive integer $r$.
\end{remark}

\section{{\wxyr Knowledge-based DRCBF}} \label{Sec:DRCBF}
% In this section, we propose a disturbance rejection CBF (DRCBF) framework for Problem \ref{problem:p1} under the worst-case disturbance and scaling method.
{\wxy
In this section, we assume the disturbance bound $\mathcal{D}$ is known, and propose a {\wxyr knowledge-based DRCBF (kDRCBF)}.
To better understand the motivation of this paper, we start with an example of adaptive cruise control (ACC) problem.
}

\subsection{Motivating Example} \label{Sec:ACC}
Consider an ACC system with the following dynamics
\begin{align} \label{eq:ACC}
&\dot{D}(t) = v_l - v_f(t) + d_u(t), \notag \\
&\dot{v}_f(t) = -\frac{1}{M} F_r(v_f(t)) + \frac{1}{M}u(t) + d_m(t),
\end{align}
    % \begin{align} \label{eq:ACC}
    % \begin{bmatrix}
    %     \dot{D} \\
    %     \dot{v}_f
    % \end{bmatrix} =
    % \begin{bmatrix}
    %     v_l - v_f \\
    %     -\frac{1}{M}F_r(v_f)
    % \end{bmatrix}
    % +
    % \begin{bmatrix}
    %     0 \\
    %     \frac{1}{M}
    % \end{bmatrix}
    % u
    % +
    % \begin{bmatrix}
    %     d_u \\
    %     d_m
    % \end{bmatrix},
    % \end{align}
where $v_l$ and $v_f$ are the velocities of the lead vehicle and the following vehicle, respectively; $D$ is the distance between them;  $M$ is the mass of the follower; $u$ is the input force to the follower; $d_u$ represents unmatched disturbances on velocity;
% such as dynamic road friction and measurement noise of the velocity; 
$d_{m}$ represents matched disturbance on acceleration;
and $F_r = f_0+f_1v_f+f_2v_f^2$ models the aerodynamic drag, with $f_0$, $f_1$ and $f_2$ being empirically determined \cite{ames2017control}.

Let $x = [D,v_f]^\top$.
The follower is required to satisfy a safety constraint described by $b(x) = D - D_{min}\geq 0$, and $D_{min}$ is the safe distance.
The IRD and DRD of $b(x)$ with respect to the system \eqref{eq:ACC} are unmatched, with $m = 2$ and $r = 1$.
For simplicity, we assume that the leader drives at a constant speed.
Then the time derivatives of $b$ along \eqref{eq:ACC} are
\begin{subequations}\label{eq:robust_h_function}
\begin{align}
    &\dot{b}(x) = v_l(t) - v_f(t) + d_{u}(t), \label{eq:robust_h_function_a}\\
    &\ddot{b}(x)  =\frac{1}{M}F_r(v_f(t)) - \frac{1}{M}u(t)+ \dot{d}_{u}(t) - d_{m}(t). \label{eq:robust_h_function_b}
\end{align}
\end{subequations}
From \eqref{eq:robust_h_function}, $d_u$ and $\dot{d}_u$ not only pollute $b$ via matched channel (see \eqref{eq:robust_h_function_b}) but also affect it from unmatched channel (see \eqref{eq:robust_h_function_a}).
Then, the HOCBF {\wxyr for the ACC system \eqref{eq:ACC}, which can be expressed as a linear combination of $b,~\dot{b}$ and $\ddot{b}$ (see \eqref{eq:hocbf_linear}),} is influenced simultaneously by $d_m$, $d_u$ and $\dot{d}_u$.
Existing robust CBFs need the bounds of $d_m$, $d_u$ and $\dot{d}_u$,
which are, however, generally intractable to obtain in practice, let alone the case when $d_u$ is non-differentiable.
This issue fails most CBF strategies, and will be addressed by our proposed DRCBFs.

\subsection{Knowledge-based DRCBF}
{\wxy
{\wxyr 
To avoid involving high-order derivatives of $d$ in HOCBFs, we recursively construct robust CBFs by replacing $d$ in \eqref{eq:hocbf_linear} with its bound $\mathcal{D}$.
For \eqref{eq:dis_sys}, the time derivative of $b(x)$ is
\begin{align*}
    \dot{b}(x)= L_f b(x) + L_g b(x)u + L_hb(x)d,
\end{align*}
where $L_h$ denote the Lie derivative along $h$ and} $L_g b = 0$ if $m > 1$.
{\wxyr Using norm operation and the fact $\lVert d\rVert\leq \mathcal{D}$, we define a lower bound of $\dot{b}(x)$ for \eqref{eq:dis_sys} in the following form}
% {\color{green} More motivations how the below is defined, as shown in the response.}
\begin{align} \label{eq:tilde_b_1}
    \bar{b}(x) = L_f b(x) +L_gb(x)u - \mathcal{D}\lVert L_h b (x)\rVert  ,
\end{align}
 % It is trivial to show that $\bar{b}(x)$ is a lower bound of the time derivative of $b(x)$ {\wxyr for system \eqref{eq:dis_sys}} with respect to $d$.
%
However, {\wxyr since the derivative of a norm function $p(r) = \lVert r\rVert$ is $\partial p(r)/\partial r = r/\lVert r\rVert$}, $\lVert L_h b(x) \rVert$ is not differentiable {\wxyr at some singular points $x\in \mathbb{R}^n$ such that $L_h b(x)=0$.}
% {\color{green} make it clear at which point.} 
{\wxyr And thus, it} cannot be used to further construct a higher-order CBF.
The fact $(\frac{1}{2\sqrt{k}}\lVert L_h b(x)\rVert - \sqrt{k}\mathcal{D})^2 \geq 0$  implies that
\begin{align} \label{eq:Young_inequality}
\frac{1}{4k}\lVert L_h b(x)\rVert^2 + k \mathcal{D}^2 \geq \mathcal{D} \lVert L_h b(x)\rVert 
\end{align} for all $k > 0$.
Compared with the term $\mathcal{D}\lVert L_h b(x)\rVert $, its upper bound $\frac{1}{4k}\lVert L_h b(x)\rVert^2 + k \mathcal{D}^2$ is differentiable, which benefits the higher-order CBF design.
Thus, we obtain a differentiable lower bound of the time derivative of $b$ {\wxyr along \eqref{eq:dis_sys}} as
\begin{align} \label{eq:tilde_b_1_modify}
    \tilde{b}(x) = L_f b(x)  +L_gb(x)u- \frac{1}{4k}\lVert L_h b(x)\rVert^2- k \mathcal{D}^2.
\end{align}
Interestingly, the introduced parameter $k$ can be adjusted to reduce conservativeness. This will be illustrated in Remark \ref{rm:para_design_drcbf}.}

% the {\wxy sufficiently differentiable function $b(x)$ of IRD $m$} {\color{red} what is the relative degree of $b(x)$?}, 
Following the above manipulation, we can recursively define a sequence of functions $\tilde{b}_i$ as
\begin{align} \label{eq:tilde_b}
    &\tilde{b}_i(x) = \tilde{w}_i(x)  - k_i \mathcal{D}^2, ~~~ i\in\mathbb{N}_{1:m-1} \nonumber \\
    &\tilde{b}_m(x) = \tilde{w}_m(x)  - k_m \mathcal{D}^2+ \beta_u(x) u,
\end{align}
where
% {\color{red}there is no $b_0$ in the above.} {\wxy The $i$ is redefined by starting from $0$.}
$\tilde{w}_i = L_f\tilde{b}_{i-1} - \frac{1}{4k_i} \lVert L_h\tilde{b}_{i-1} \rVert^2$, $\tilde{b}_0(x) = b(x)$, $k_i$ is a positive parameter and $\beta_u = L_g\tilde{b}_{m-1}$. 
% {\wxy Here I don't discuss the definition of $\tilde{w}_i$ because I provide the transformation process in \eqref{eq:tilde_b_1}. How do you think?~~:)}
%

The following lemma shows that the form of \eqref{eq:tilde_b} recursively defines a lower bound of the time derivative of $\tilde{b}_i$.

% The following lemma shows that, , we iteratively find the lower bound of $\dot{b}$ while preventing $d$ from polluting the CBF conditions.

\begin{lemma} \label{lm:b_i>b_j}
    {\wxyr Consider the system \eqref{eq:dis_sys}. 
    Let the functions $\tilde{b}_i$, $i\in \mathbb{N}_{1:m}$ be defined by \eqref{eq:tilde_b}.}
    Then the following inequality
    \begin{align*}
        \dot{\tilde{b}}_{i-1}(x{\wxyr (t)}) \geq \tilde{b}_{i}(x{\wxyr (t)}), ~~~~ i \in \mathbb{N}_{1:m}
    \end{align*}
    holds for all $t \geq 0$.
\end{lemma}

\begin{proof}
    The time derivative of $\tilde{b}_i(x)$ along \eqref{eq:dis_sys} is
    \begin{align} \label{eq:tilde_b_dot}
        \dot{\tilde{b}}_{i-1}(x) &=   L_f\tilde{b}_{i-1}(x) +  L_h\tilde{b}_{i-1}(x) d,  ~~~ i\in\mathbb{N}_{1:m-1} \nonumber\\
        \dot{\tilde{b}}_{m-1}(x) &=  L_f\tilde{b}_{m-1}(x)  + L_h\tilde{b}_{m-1}(x) d + \beta_u(x)u .
    \end{align}
    Noting that $(\frac{1}{2\sqrt{k_i}}L_h \tilde{b}_{i-1}^\top + \sqrt{k_i}d)^2 \geq 0$ holds for arbitrary positive $k_i$,  one can obtain a lower bound of \eqref{eq:tilde_b_dot} as 
    %
    % Using Young's inequality {\color{red} can you write this in more details on the following equation? e.g., there is no time derivative of $\tilde{b}_{i-1}$ in the equation below}~
    % {\wxy I rewrite the proof to make it more convenience for reading.}
    % % {\wxy Sure. The idea using in this part is very similar to your 2024 CDC paper on rejecting stochastic disturbance; yet the scaling we use is based on young's inequality, not the log function. Here I rewrite the proof to make it more convenience for reading.}
    % \begin{align*}
    %     L_h \tilde{b}_{i-1}(x) d \geq - \frac{1}{4k_i}\lVert L_h \tilde{b}_{i-1}(x)\rVert^2 - k_i \lVert d\rVert^2,
    % \end{align*}
    % one can always obtain the lower bound of the time derivative of $\tilde{b}_{i} ~ \forall i \in \mathbb{N}_{0:m-1}$ as
    % {\color{red}I am stopping here}
    \begin{align} \label{eq:tilde_b_lowerbound}
        \dot{\tilde{b}}_{i-1}(x) &\geq   \tilde{w}_i(x) - k_i \lVert d(t)\rVert^2,~~~ i\in\mathbb{N}_{1:m-1} \nonumber\\
        \dot{\tilde{b}}_{m-1}(x) &\geq  \tilde{w}_m(x) + \beta_u(x) u - k_m \lVert d(t)\rVert^2.
    \end{align}
    Recalling \eqref{eq:tilde_b}, {\wxyr one can further put \eqref{eq:tilde_b_lowerbound} as}
    \begin{align*}
        \dot{\tilde{b}}_{i-1}{\wxyr (x)} &\geq   \tilde{b}_i{\wxyr (x)} -  k_i(\lVert d(t)\rVert^2 - \mathcal{D}^2),~~~ i\in\mathbb{N}_{1:m}.
    \end{align*}
    {\wxyr Since $\lVert d(t)\rVert \leq \mathcal{D}$, one has} $\dot{\tilde{b}}_{i-1}(x) \geq  \tilde{b}_{i}(x)$, $ \forall i \in \mathbb{N}_{1:m}$.
\end{proof}

By replacing $b^{(i)}$ in \eqref{eq:hocbf_linear} with $\tilde{b}_i$ in \eqref{eq:tilde_b}, we define a series of functions {\wxyr in the following form}
% \begin{align} \label{eq:dohcbf_functions}
%         \tilde{\varphi}_i(x) = &  \tilde{b}_i(x)  +\sum_{j=0}^{i-1} c_j^{i}  \tilde{b}_j(x), ~~ i \in \mathbb{N}_{1:m-1} \nonumber \\
%     \tilde{\varphi}_m(x) = &  \tilde{w}_m(x) - k_m \mathcal{D}^2 +\beta_u(x) u +\sum_{j=0}^{m-1} c_j^m \tilde{b}_j(x).
% \end{align}
\begin{align} \label{eq:dohcbf_functions}
        \tilde{\varphi}_i(x) =  \tilde{b}_i(x)  +\sum_{j=0}^{i-1} c_j^{i}  \tilde{b}_j(x), ~~ i \in \mathbb{N}_{1:m}.
\end{align}
Further define a sequence of sets associated with \eqref{eq:dohcbf_functions} as
\begin{align} \label{eq:dohcbf_sets}
    \mathscr{C}_i = \{ x\in \mathbb{R}^n:~ \tilde{\varphi}_{i-1} (x)\geq 0\}, ~~ i\in\mathbb{N}_{1:m},
\end{align}
where $\tilde{\varphi}_{0} (x) = b(x)$.
%
% Let $\bar{\mathscr{C}}=\cap_{i\in\mathbb{N}_{1:m}} \mathscr{C}_i$.
% %
% Obviously, $\bar{\mathscr{C}} \subset \mathscr{C}$.
%
{\hw It is worth noting that, by using \eqref{eq:dohcbf_functions}, we successfully avoid the {\wxyr emergence} of the disturbance $d$ and its derivatives of different orders on the CBF {\wxyr constraints} \eqref{eq:hocbf_linear}.}

Now, we {\wxyr introduce} the definition of {\wxyr knowledge-based disturbance rejection control barrier function}.
\begin{definition} \label{def:DRCBF}
    {\wxyr (kDRCBF)
     Let the functions $\tilde{\varphi}_i$ and the sets $\mathscr{C}_i$, $i\in\mathbb{N}_{1:m}$ be defined by \eqref{eq:dohcbf_functions} and \eqref{eq:dohcbf_sets}, respectively.}
    Let $b(x)$ be an HOCBF of {\wxyr order} $m$ for the nominal system \eqref{eq:nominal_sys}. 
    Then $b(x)$ is a {\wxyr knowledge-based DRCBF of order} $m$ for the disturbed system \eqref{eq:dis_sys} if there exist  positive {\wxyr constants} $p_1, p_2, \cdots , p_m$ such that $x(0)\in \bar{\mathscr{C}}{\wxyr :=\cap_{i\in\mathbb{N}_{1:m}} \mathscr{C}_i}$ and 
    % \begin{align} \label{eq:DRCBF}
    %     \sup_{u\in \mathbb{R}^p} \bigg\{ c_0^m  b(x)+\sum_{j=1}^m  c_j^m \tilde{w}_j(x)  + \beta_u(x) u\bigg\} \geq \bar{k} \mathcal{D}^2 ,
    % \end{align}
    \begin{align} \label{eq:DRCBF}
        \sup_{u\in \mathbb{R}^p} \bigg\{ \tilde{w}_m(x) +\beta_u(x) u+\sum_{j=0}^{m-1} c_j^m \tilde{b}_j(x)\bigg\} \geq k_m \mathcal{D}^2
    \end{align}
    for all $x \in \bar{\mathscr{C}}$.
\end{definition}

% Instead of using HOCBF, we use \eqref{eq:dohcbf_functions} as the safety certification for disturbed system \eqref{eq:dis_sys}.

The following theorem illustrates that any Lipschitz continuous controller $u$ satisfying \eqref{eq:DRCBF} can {\wxyr robustly render the set $\bar{\mathscr{C}}$ forward invariant,} thus guaranteeing that $b(x(t))\geq 0,~\forall t\geq 0$.

\begin{thm}\label{thm:t1}
    Let $b(x)$ be a {\wxyr kDRCBF of order} $m$ for system \eqref{eq:dis_sys} and the sets $\mathscr{C}_i$, $i\in\mathbb{N}_{1:m}$ be {\wxyr defined by \eqref{eq:dohcbf_sets}.
    If $x(0) \in \bar{\mathscr{C}}$,}
    then any Lipschitz continuous controller $u$ satisfying 
    \begin{align} \label{eq:DRCBF_condition}
       \tilde{w}_m(x) +\beta_u(x) u+\sum_{j=0}^{m-1} c_j^m \tilde{b}_j(x)\geq k_m \mathcal{D}^2
    \end{align}
     renders $\bar{\mathscr{C}}$ forward invariant {\wxyr for \eqref{eq:dis_sys}, and $b(x(t)) \geq 0,\forall t\geq0$.}
\end{thm}

\begin{proof}
First, we introduce a sequence of functions
\begin{align}\label{eq:thm1_proof:1}
    \varphi_i(x) = \dot{\tilde{\varphi}}_{i-1}(x) + p_{i}\tilde{\varphi}_{i-1}(x)
\end{align}
with $\varphi_0(x) = \tilde{\varphi}_ 0(x)$.
Considering $\tilde{b}_0(x) = b(x)$, the definition \eqref{eq:dohcbf_functions} and Lemma \ref{lm:b_i>b_j}, {\wxyr one can show that}
\begin{align}\label{eq:thm1_proof:2}
    \varphi_1 = & \dot{\tilde{b}}_0 +p_1 \tilde{b}_0  \nonumber \\
    \geq & \underbrace{\tilde{b}_1 + p_1 \tilde{b}_0}_{ \tilde{b}_1 + c_0^1 \tilde{b}_0}  \equiv \tilde{\varphi}_1,  \nonumber\\
    \varphi_2 \geq & \underbrace{ \tilde{b}_2 + (p_1 + p_2)\tilde{b}_1 + p_2p_1\tilde{b}_0}_{ \tilde{b}_2 + c_1^2\tilde{b}_1 + c_0^2 \tilde{b}_0}   \equiv \tilde{\varphi}_2,  \nonumber\\
    \vdots &  \nonumber\\
    \varphi_m \geq & \tilde{b}_m + \cdots + c_1^m\tilde{b}_1 + c_0^m \tilde{b}_0 \equiv \tilde{\varphi}_m.
\end{align}

{\wxyr By \eqref{eq:dohcbf_functions}, one has $\tilde{\varphi}_{m}(x)\geq 0$ under} any Lipschitz continuous controller $u$ satisfying \eqref{eq:DRCBF_condition}.
Then \eqref{eq:thm1_proof:2} implies $ \varphi_{m}(x) \geq 0$, which further implies  that $\dot{\tilde{\varphi}}_{m-1}$ by \eqref{eq:thm1_proof:1}. 
{\wxyr According to Lemma 1 in \cite{xiao2021high}}, we have $\tilde{\varphi}_{m-1}(x(t)) \geq 0,~ \forall t\geq0$ since $\tilde{\varphi}_{m-1}(x(0)) \geq 0$.
Again, {\wxyr by \eqref{eq:thm1_proof:1} and \eqref{eq:thm1_proof:2},} $\tilde{\varphi}_{m-1}(x(t)) \geq 0$ implies $ \varphi_{m-1}(x(t))  \geq 0$, and therefore $\tilde{\varphi}_{m-2}(t) \geq 0, \forall t\geq 0$ since $\tilde{\varphi}_{m-2}(x(0)) \geq 0$ and $\dot{\tilde{\varphi}}_{m-2}(x)+p_{m-1}\tilde{\varphi}_{m-2}(x) \geq 0$.
Iteratively, one can conclude that $x(t) \in \mathscr{C}_i, ~ \forall i \in \mathbb{N}_{1:m}, ~\forall t \geq 0$.
Therefore, {\wxyr the set} $\bar{\mathscr{C}}$ is forward invariant {\wxyr for \eqref{eq:dis_sys}}.
Recalling $\bar{\mathscr{C}} \subset \mathscr{C}$, we have $b(x(t))\geq0,~ \forall t\geq 0$.
\end{proof}

\subsection{Optimal Control with {\wxyr kDRCBF}}
{\wxyr It has been shown} that {\wxyr kDRCBF} can guarantee the safety of system {\wxyr \eqref{eq:dis_sys}}, i.e., $b(x{\wxyr (t)})\geq 0,~\forall t \geq 0$.
To stabilize system \eqref{eq:dis_sys}, we resort to the input-to-state stability CLF (ISS-CLF).
\begin{definition}
    (ISS-CLF \cite{sontag1995on})
    A continuously differentiable positive definite function $V:\mathbb{R}^n \rightarrow\mathbb{R}$ is an ISS-CLF function for \eqref{eq:dis_sys} if there exist class $\mathcal{K}_\infty$ functions $\tilde{\alpha}_1$ and $\tilde{\alpha}_2$ such that 
    \begin{align} \label{eq:ISS-CLF}
        \inf_{u\in \mathbb{R}^p} \{ L_f V + L_gVu + L_h Vd\} \leq -\tilde{\alpha}_1(V) + \tilde{\alpha}_2(\lVert d\rVert)
    \end{align}
    for all $x \in \mathbb{R}^n$ and $d \in \mathbb{R}^q$.
\end{definition}

Let $V(x)$ be an ISS-CLF for system \eqref{eq:dis_sys}.
Then any Lipschitz continuous controller $u$ satisfying $L_fV+L_gVu \leq -\sigma V$ for a {\wxyr positive $\sigma$} can guarantee the input-to-state stability (Definition 2.1, \cite{sontag1995on}) of system \eqref{eq:dis_sys}.
To achieve both stability and safety, we reformulate Problem \ref{problem:p1} as the following {\wxyr quadratic program}
\begin{align*}
    &\operatorname*{min}_{(u,\delta) \in \mathbb{R}^p \times \mathbb{R}} 
   ~~~ J(u) + \rho {\wxyr \delta}^2~~~~~~~~~~~~\\
    &~~~~~\text{s.t. }~~~ L_fV(x)+L_gV(x)u \leq -\sigma V(x)+{\wxyr \delta}\\
    &~~~~~~~~~~~~ \tilde{w}_m(x) +\beta_u(x) u+\sum_{j=0}^{m-1} c_j^m \tilde{b}_j(x)\geq k_m \mathcal{D}^2,
\end{align*}
where ${\wxyr \delta}$ is a slack variable {\wxyr to relax ISS-CLF condition} for guaranteeing the quadratic program to be feasible and {\wxyr $\rho$ is a penalty factor to make $\delta$ in the solution as small as possible}.

\begin{remark} \label{rm:para_design_drcbf}
    (\textbf{Parameter design of {\wxyr kDRCBF}})
    In {\wxyr kDRCBF}, we introduce a set of parameters $k_1,k_2,\cdots,k_m$ for control design.
    {\wxyr Now we show} how to design these parameters to reduce the conservativeness caused by the {\wxyr over-approximation of $d$.}
    This is achieved by finding the minimum of $\frac{1}{4k_i}\lVert L_h \tilde{b}_{i-1}\rVert^2 + k_i \mathcal{D}^2$ with respect to $k_i$.
    If $\lVert L_h \tilde{b}_{i-1}\rVert$ is upper bounded by a known positive value $\eta_i$, then $\frac{1}{4k_i}\lVert L_h \tilde{b}_{i-1}\rVert^2 + k_i \mathcal{D}^2$ is upper bounded by $\frac{1}{4k_i}\eta_i^2 + k_i \mathcal{D}^2$.
    Let $\varrho_i(k_i) = \frac{1}{4k_i}\eta_i^2 + k_i \mathcal{D}^2$.
   Then the minimum of $\varrho(k_i)$ is achieved at $k^*_i = \frac{\eta_i}{2\mathcal{D}}$, which solves $\frac{\partial \varrho_i(k_i)}{\partial k_i} = 0$.
    Based on these manipulations, the parameters of DRCBF  can be chosen as $k^*_1 = \frac{\eta_1}{2\mathcal{D}}, k^*_2 = \frac{\eta_2}{2\mathcal{D}},\cdots,k^*_m = \frac{\eta_m}{2\mathcal{D}}$. These parameters can provide the least conservativeness of the DRCBF.
    By doing so, we can find the least-conservative upper bound of CBF under the worst-case disturbance as
    \begin{align*}
        L_f\tilde{b}_{i-1} - &\frac{1}{4k_i} \lVert L_h\tilde{b}_{i-1} \rVert^2 - k_i\mathcal{D}^2\\
       & ~~~~~~\leq L_f\tilde{b}_{i-1} - \frac{1}{4k^*_i} \lVert L_h\tilde{b}_{i-1} \rVert^2 - k^*_i\mathcal{D}^2
    \end{align*}
    for arbitrary $(k_1,k_2,\cdots,k_m) \in \mathbb{R}_{> 0} \times \mathbb{R}_{> 0}\times,\cdots ,\times\mathbb{R}_{> 0}$.
    
\end{remark}

% To design a robust controller for , we establish the following sliding surface 
% \begin{align*}
%     s(x) = x(t) - x(0) - \int_0^t (f(x)+g(x)u)d\tau
% \end{align*}

\section{{\wxyr Reciprocal-compensated} DRCBF} \label{Sec:aDRCBF}
Although the proposed {\wxyr kDRCBF} can enforce robust safety, 
%
% However, this result is developed under the worst-case disturbance, thereby leads to undesirable conservativeness.
% %
% In this section, we introduce an adaptive DRCBF (aDRCBF) to address Problem \ref{problem:p1} while reducing conservativeness of DRCBF.
this safety guarantee is developed under the assumption that $\mathcal{D}$, i.e., {\wxyr the disturbance bound}, is known, which is not the case in many practical applications.
Now we shall introduce a new class of DRCBF, called {\wxyr reciprocal-compensated DRCBF (rDRCBF)}, to address Problem \ref{problem:p1} without knowing $\mathcal{D}$.

\subsection{{\wxyr Reciprocal-compensated} DRCBF}
%
% We first give a lower bound of the time derivative of $b(x)$ with respect to $d$ as 
% \eqref{eq:tilde_b_1}:
% \begin{align} \label{eq:tilde_psi_1}
%     \tilde{\psi}_1(x) = L_f \tilde{\psi}_0(x) - \mathcal{D}\lVert L_h\tilde{\psi}_0(x) \rVert,
% \end{align}
% where $\tilde{\psi}_0(x) = b(x)$ and $L_g \tilde{\psi}_0 = 0$ if $m >1$.
% %
% To obtain a differentiable over-approximation of $\mathcal{D}\lVert L_h\tilde{\psi}_0(x) \rVert$, we modify \eqref{eq:tilde_psi_1} by following \eqref{eq:tilde_b_1_modify} in the form
% \begin{align} \label{eq:tilde_psi_1_mod}
%     \tilde{\psi}_1(x) = L_f \tilde{\psi}_0(x) - \frac{1}{4k_1}\lVert L_h\tilde{\psi}_0(x) \rVert^2 - k_1 \mathcal{D}^2,
% \end{align}
% where $k_1>0$ is introduced to provide more freedom in CBF design.

% Since $\mathcal{D}$ is bounded, there always exists a continuously differentiable function $\Gamma(x)$, independent of $\mathcal{D}$, such that $\Gamma(x) \geq \mathcal{D}^2$ when $x$ approaches the boundary of the safe set.
{\wxyr When $\mathcal{D}$ in \eqref{eq:tilde_b_1} is unknown, we replace $\mathcal{D}$ with a sufficiently differentiable function $\Gamma(x)$ which is} independent of $\mathcal{D}$ {\wxyr and tends to infinity as} $x$ approaches the boundary of the safe set, {\wxyr i.e., $\Gamma(x) \rightarrow \infty$ as $b(x) \rightarrow 0$}.
Then similar to \eqref{eq:tilde_b_1_modify}, we define a function 
\begin{align} \label{eq:tilde_psi_1_mod_mod}
    \tilde{\psi}(x) = L_f b(x) +L_gb(x)u -\frac{1}{4k} \lVert L_hb(x) \rVert^2-k \Gamma(x),
\end{align}
% which provides a lower bound of $\dot{b}$ when {\wxyr $b$ decreases to zero.}
{\wxyr Substituting \eqref{eq:tilde_b_1_modify} into \eqref{eq:tilde_psi_1_mod_mod} yields
\begin{align} \label{eq:use}
    \tilde{\psi}(x) = \tilde{b}(x) -k(\Gamma(x) - \mathcal{D}^2).
\end{align}
For any positive number $\mathcal{D}$, the property of $\Gamma(x)$ guarantees that there exists a positive number $\varepsilon$ such that $\Gamma(x) \geq \mathcal{D}^2$ whenever $b(x) \leq \varepsilon$.
Thus, according to \eqref{eq:use} and Lemma \ref{lm:b_i>b_j}, $\tilde{\psi}(x)$ provides a lower bound of $\dot{b}(x)$ for system \eqref{eq:dis_sys} when $x$ gets close to the boundary of safe set.
}
{\wxyr Following a similar development as in the construction of kDRCBFs}, we recursively define a series of functions as
\begin{align} \label{eq:tilde_psi}
    &\tilde{\psi}_i(x) = \tilde{\pi}_i(x) - k_i \Gamma_{i-1}(x) ,~~~~ i\in\mathbb{N}_{1:m-1} \nonumber \\
    &\tilde{\psi}_m(x) = \tilde{\pi}_m(x)- k_m \Gamma_{m-1}(x) + \tilde{\beta}_u(x) u,
\end{align}
where $\tilde{\pi}_i =L_f \tilde{\psi}_{i-1} - \frac{1}{4k_i}\lVert L_h \tilde{\psi}_{i-1} \rVert$, $\tilde{\psi}_0(x) = b(x)$, $k_i$ is a positive parameter, $\tilde{\beta}_u = L_g\tilde{\psi}_{m-1}$ and $\Gamma_i:\mathbb{R}^n \rightarrow \mathbb{R}$ is {\wxyr a sufficiently} differentiable function to be designed later.

By replacing $b^{(i)}$ in \eqref{eq:hocbf_linear} with $\tilde{\psi}_i$ in \eqref{eq:tilde_psi}, we define a series of functions $\tilde{\phi}_i$ as
% \begin{align} \label{eq:adaptve_dohcbf_functions}
%     \tilde{\phi}_i(x) = &  \tilde{\psi}_i(x) +\sum_{j=0}^{i-1} c_j^i  \tilde{\psi}_j(x), ~~ i \in \mathbb{N}_{1:m-1} \nonumber \\
%     \tilde{\phi}_m(x) = & \tilde{\pi}_m(x)  - k_m \Gamma_{m-1}(x) + \tilde{\beta}_u(x) u +\sum_{j=0}^{m-1} c_j^m \tilde{\psi}_j(x)
% \end{align}
\begin{align} \label{eq:adaptve_dohcbf_functions}
    \tilde{\phi}_i(x) = &  \tilde{\psi}_i(x) +\sum_{j=0}^{i-1} c_j^i  \tilde{\psi}_j(x), ~~ i \in \mathbb{N}_{1:m}
\end{align}
with $\tilde{\phi}_0(x) = \tilde{\psi}_0(x)$.
We then define a sequence of sets $\mathcal{C}_i, \forall i\in\mathbb{N}_{1:m}$ associated with \eqref{eq:adaptve_dohcbf_functions} in the form 
\begin{align} \label{eq:adaptive_dohcbf_sets}
    \mathcal{C}_i = \{ x\in \mathbb{R}^n ~:~ \tilde{\phi}_{i-1} (x)\geq 0\}, \nonumber \\
    \partial \mathcal{C}_i = \{ x\in \mathbb{R}^n ~:~ \tilde{\phi}_{i-1} (x) = 0\},\nonumber \\
    \operatorname{Int}(\mathcal{C}_i) = \{ x\in \mathbb{R}^n ~:~ \tilde{\phi}_{i-1} (x) > 0\}.
\end{align}
To over-approximate the bound of the disturbance, i.e., $\mathcal{D}$, we design $\Gamma_{i}$ in the following form
\begin{align*}
    \Gamma_i(x) = r_iB(\tilde{\phi}_{i}(x) ), ~~  i \in \mathbb{N}_{0:m-1}
\end{align*}
where {\wxyr $r_i$ is a positive number} and $B:\mathbb{R}_{>0} \rightarrow \mathbb{R}_{\geq0}$ is a sufficiently differentiable {\wxyr reciprocal-like} function such that
\begin{align} \label{eq:B_property}
    1/\breve{\alpha}_1(\tilde{\phi}_i) \leq B(\tilde{\phi}_i) \leq 1/\breve{\alpha}_2(\tilde{\phi}_i),~~\tilde{\phi}_i\in \mathbb{R}_{>0}
\end{align}
holds for some class $\mathcal{K}_\infty$ functions $\breve{\alpha}_1, \breve{\alpha}_2$.
Obviously, $B(\tilde{\phi}_{i-1})$ grows rapidly to infinity as $x$ approaches $\partial \mathcal{C}_i$, i.e., the boundary of $\mathcal{C}_i$ (cf. reciprocal CBFs in \cite{ames2017control}).
In other words, $\mathcal{D}$ will be upper bounded by $\Gamma_{i-1}$ when $x$ gets close enough to $\partial \mathcal{C}_i$.
And then, we can use this property to enforce robust safety without knowing $\mathcal{D}$.
{\wxyr A valid candidate for $B$ is $1/\tilde{\phi}_{i}$.}

Before proceeding, we present the following lemma.
\begin{lemma} \label{lm:b_i>b_j_adaptive}
    Consider the system \eqref{eq:dis_sys}.
    Let the functions $\tilde{\psi}_i$ and the sets $\mathcal{C}_i$ be defined by \eqref{eq:tilde_psi} and \eqref{eq:adaptive_dohcbf_sets}, respectively.
    %
    % Let $c_j^i~~ \forall j\in \mathbb{N}_{0:i},i\in\mathbb{N}_{1:m} $ be the positive parameters of polynomial $c^i_i s^i + c^i_{i-1} s^{i-1} + \cdots + c^i_0$ with the eigenvalues of $\tilde{p}_{i}, \tilde{p}_{i-1}, \cdots, \tilde{p}_1$, where $\tilde{p}_i > 0 ~~ \forall i \in \mathbb{N}_{1:m}$.
%
    Let $\bar{\mathcal{C}}=\cap_{i\in\mathbb{N}_{1:m}} \operatorname{Int}(\mathcal{C}_i)$.  
    There always exists a neighborhood of $\partial \mathcal{C}_l$, denoted by $\mathcal{N}(\partial \mathcal{C}_l)$, such that $\forall l \in \mathbb{N}_{1:m}, i \in \mathbb{N}_{l:m}$
    \begin{align} \label{eq:lm_b_i_adaptive_inequality}
         \sum_{j=0}^{i-1} c_j^{i-1} (\dot{\tilde{\psi}}_j(x) - \tilde{\psi}_{j+1}(x) ) {\wxyr \geq \mathcal{H}}, ~ \forall x \in \mathcal{N}(\partial \mathcal{C}_l) \cap \bar{\mathcal{C}},
    \end{align}
   where {\wxyr $\mathcal{H}$ is a positive constant and} $c_s^s = 1$ for all $s \in \mathbb{N}_{1:m}$.
\end{lemma}

\begin{proof}
    Similar to \eqref{eq:Young_inequality}, the time derivative of $\tilde{\psi}_i(x)$ along \eqref{eq:dis_sys} can be lower bounded as 
    \begin{align} \label{eq:tilde_psi_lower_bound}
        &\dot{\tilde{\psi}}_{i-1}(x) \geq \tilde{\pi}_{i}(x) - k_{i} \lVert d\rVert^2 ,~~i \in \mathbb{N}_{1:m-1} \nonumber\\
        &\dot{\tilde{\psi}}_{m-1}(x) \geq \tilde{\pi}_{m}(x) - k_{m} \lVert d\rVert^2  + \tilde{\beta}_{u}(x) u.
    \end{align}
    %
    % One can rewrite \eqref{eq:tilde_psi} as 
    % \begin{align} \label{eq:tilde_psi_rewrite}
    %     &k_{i}\Gamma_{i-1}(x) = \tilde{\pi}_{i}(x) - \tilde{\psi}_{i}(x),~~~ i \in \mathbb{N}_{1:m-1}\notag\\
    %     &k_{m}\Gamma_{m-1}(x) = \tilde{\pi}_{m}(x) - \tilde{\psi}_{m}(x) + \tilde{\beta}_{u}(x)u.
    % \end{align} 
    % By adding zero and substituting \eqref{eq:tilde_psi_rewrite} into \eqref{eq:tilde_psi_lower_bound}, 
  Recalling \eqref{eq:tilde_psi},  the inequality \eqref{eq:tilde_psi_lower_bound} can be further put as 
    \begin{align} \label{eq:tilde_psi_lower_bound_1}
        \dot{\tilde{\psi}}_{i-1}(x) \geq & \tilde{\psi}_{i}(x) + (\tilde{\pi}_{i}(x)  - \tilde{\psi}_{i}(x) - k_{i} \lVert d \rVert^2) \notag\\
        = & \tilde{\psi}_{i}(x) + k_{i}( \Gamma_{i-1}(x)- \lVert d \rVert^2),~~~~~~ i \in \mathbb{N}_{1:m-1}\notag\\
        \dot{\tilde{\psi}}_{m-1}(x) \geq & \tilde{\psi}_{m}(x) + (\tilde{\pi}_{m}(x)  - \tilde{\psi}_{m}(x)+\tilde{\beta}_u(x) u- k_{m} \lVert d \rVert^2) \notag \\
        = & \tilde{\psi}_{m}(x) + k_{m}( \Gamma_{m-1}(x)- \lVert d \rVert^2).
    \end{align}
    Considering $\Gamma_j =r_j B(\tilde{\phi}_j)$ and the fact that $r_j,~k_ j,~c_j^i > 0$ and $B \geq 0, ~\forall x \in \bar{\mathcal{C}}$, the inequality \eqref{eq:tilde_psi_lower_bound_1} further implies  
    \begin{align*}
        \sum_{j=0}^{i-1} c_j^{i-1} (\dot{\tilde{\psi}}_j(x) - \tilde{\psi}_{j+1}(x) ) \geq& \sum_{j=0}^{i-1}  k_{j+1}c_j^{i-1}(r_j B(\tilde{\phi}_j) - \lVert d \rVert^2) \\
         \geq&  k_l r_{l-1}  c_{l-1}^{i-1} B(\tilde{\phi}_{l-1})  - \breve{c}_{i-1}\lVert d \rVert^2
    \end{align*}
    for any $l \in \mathbb{N}_{1:i}$, where $\breve{c}_{i-1} = \sum_{j=0}^{i-1} k_{j+1} c_j^{i-1}$.
    It follows from \eqref{eq:B_property} that $B(\tilde{\phi}_{l-1}(x))$ approaches infinity as $x\rightarrow \partial \mathcal{C}_l$.
    Given arbitrary positive $\mathcal{O} \in \mathbb{R}_{>0}$, there always exists a small neighborhood $\mathcal{N}(\partial \mathcal{C}_l)$ such that $B(\tilde{\phi}_{l-1}(x)) \geq \mathcal{O}$ for all $x \in \mathcal{N}(\partial \mathcal{C}_l)\cap \bar{\mathcal{C}}$.
    Recalling that $\lVert d \rVert\leq \mathcal{D}$, and $\mathcal{D}$, $c_j^{i-1},~k_j$ are positive real numbers, one has that $\breve{c}_{i-1}\lVert d\rVert^2$ is bounded.
    Selecting $\mathcal{O} = \frac{\breve{c}_{i-1}\mathcal{D}^2{\wxyr +\mathcal{H}}}{k_l r_{l-1}  c_{l-1}^{i-1}}$ yields $k_l r_{l-1}  c_{l-1}^{i-1} B(\tilde{\phi}_{l-1})  - \breve{c}_{i-1}\lVert d \rVert^2 {\wxyr \geq \mathcal{H}}$ for all $x \in \mathcal{N}(\partial \mathcal{C}_i)\cap \mathcal{C}_i$, which completes the proof.
\end{proof}

% The following theorem shows that {\wxyr rDRCBF guarantees the existence of a control input $u$ that renders} the set $\bar{\mathcal{C}}$ forward invariant without knowing {\wxyr any information of $d$.}

% The following theorem shows that any Lipschitz continuous controller $u$ satisfying \eqref{eq:adaptive_DRCBF} can render the set $\bar{\mathcal{C}}$ forward invariant without prior knowledge of disturbance bound $\mathcal{D}$.
 
Now, we propose the definition of {\wxyr reciprocal-compensated disturbance rejection control barrier function.}
\begin{definition} \label{def:adaptive_DRCBF}
   {\wxyr(rDRCBF) Let the functions $\tilde{\phi}_i$ and the sets $\mathcal{C}_i$, $i\in\mathbb{N}_{1:m}$ be defined by \eqref{eq:adaptve_dohcbf_functions} and \eqref{eq:adaptive_dohcbf_sets}, respectively.}
    Let $b(x)$ be an HOCBF {\wxyr of order $m$} for the nominal system \eqref{eq:nominal_sys}.
    Then $b(x)$ is {\wxyr a reciprocal-compensated DRCBF of order $m$ for the disturbed system} \eqref{eq:dis_sys} if there exist positive values $p_1, p_2, \cdots,p_m$ such that $x(0)\in \bar{\mathcal{C}}:=\cap_{i\in\mathbb{N}_{1:m}} \operatorname{Int}(\mathcal{C}_i)$ and 
    \begin{align*} %\label{eq:adaptive_DRCBF}
        \sup_{u\in \mathbb{R}^p} \bigg\{\tilde{\pi}_m(x) + \tilde{\beta}_u(x) u +\sum_{j=0}^{m-1} c_j^m \tilde{\psi}_j(x) \bigg\} \geq k_m \Gamma_{m-1}(x)
    \end{align*}
    for all $x \in \bar{\mathcal{C}}$.
\end{definition}

\begin{thm} \label{thm:t2}
    % Consider the system \eqref{eq:dis_sys}.
    {\wxyr Let $b(x)$ be an rDRCBF of order $m$ for system \eqref{eq:dis_sys} and the sets $\mathcal{C}_i,~i\in\mathbb{N}_{1:m}$ be defined by \eqref{eq:adaptve_dohcbf_functions} and \eqref{eq:adaptive_dohcbf_sets}. }
    {\wxyr If $x(0) \in \bar{\mathcal{C}}$}, then any Lipschitz continuous controller $u$ satisfying
    \begin{align}\label{eq:aDRCBF_condition}
        \tilde{\pi}_m(x) + \tilde{\beta}_u(x) u +\sum_{j=0}^{m-1} c_j^m \tilde{\psi}_j(x) \geq k_m \Gamma_{m-1}(x)
    \end{align}
    renders $\bar{\mathcal{C}}$ forward invariant {\wxyr for \eqref{eq:dis_sys}}, and  $b(x(t))> 0, ~\forall t \geq 0$.
\end{thm}

\begin{proof}
First, we introduce a sequence of functions
\begin{align}\label{eq:thm2_proof:1}
    \phi_i(x) = \dot{\tilde{\phi}}_{i-1}(x) + p_{i}\tilde{\phi}_{i-1}(x), ~~{\wxyr i\in\mathbb{N}_{1:m}}
\end{align}
with $\phi_0 = \tilde{\phi}_ 0$.
Note that $\Gamma_{i-1}$ is a function of $\tilde{\phi}_{i-1}$ and $\tilde{\phi}_i = \sum_{j=0}^i c_j^i  \tilde{\psi}_j$.
Then $\Gamma_{i-1}$ always has higher input relative degree than $\tilde{\phi}_{i}$, which implies that no additional input term will be introduced by $\Gamma_{i-1}$.
Substituting \eqref{eq:adaptve_dohcbf_functions} into \eqref{eq:thm2_proof:1}, {\wxyr one has} 
\begin{align}\label{eq:thm2_proof:2}
    \phi_1 = & \dot{\tilde{\psi}}_0 +p_1 \tilde{\psi}_0  \nonumber \\
    = & \underbrace{\tilde{\psi}_1 +p_1 \tilde{\psi}_0}_{\tilde{\psi}_1 +c_0^1\tilde{\psi}_0} + c_0^0(\dot{\tilde{\psi}}_0-\tilde{\psi}_1), \nonumber \\
    \phi_2 = &\underbrace{ \tilde{\psi}_2 + (p_1 + p_2)\tilde{\psi}_1 + p_2p_1\tilde{\psi}_0}_{\tilde{\psi}_2 + c_1^2\tilde{\psi}_1 + c_0^2 \tilde{\psi}_0}+ \sum_{j=0}^{1} c_j^{1} (\dot{\tilde{\psi}}_j - \tilde{\psi}_{j+1} ), \nonumber \\
   \vdots &  \nonumber \\
   \phi_{m} = &\sum_{j=0}^{m}c^m_j \tilde{\psi}_j + \sum_{j=0}^{m-1} c_j^{m-1} (\dot{\tilde{\psi}}_j - \tilde{\psi}_{j+1} ) .
\end{align}
Noting that $\tilde{\phi}_i = \sum_{j=0}^i c_j^i  \tilde{\psi}_j$, one has
\begin{align} \label{eq:thm_2_proof:2}
    \phi_{i} = \tilde{\phi}_i + \sum_{j=0}^{i-1} c_j^{i-1} (\dot{\tilde{\psi}}_j(x) - \tilde{\psi}_{j+1}(x) ).
\end{align}
Combining \eqref{eq:thm2_proof:1} and \eqref{eq:thm_2_proof:2}, one can obtain that 
\begin{align} \label{eq:thm_2_proof:3}
    \dot{\tilde{\phi}}_{i-1} + p_i\tilde{\phi}_{i-1} = \tilde{\phi}_i + \sum_{j=0}^{i-1} c_j^{i-1} (\dot{\tilde{\psi}}_j(x) - \tilde{\psi}_{j+1}(x)).
\end{align}

%
% This function \eqref{eq:thm_2_proof:3} closely resembles the standard high-oder CBF defined in \eqref{eq:hocbf}, with the key distinction being the inclusion of an additional term $\Delta_i:=\sum_{j=0}^{i-1} c_j^{i-1} (\dot{\tilde{\psi}}_j - \tilde{\psi}_{j+1})$.
% Let $\Delta_i=\sum_{j=0}^{i-1} c_j^{i-1} (\dot{\tilde{\psi}}_j - \tilde{\psi}_{j+1})$.
% %
% If $\Delta_i(t) \geq 0, ~\forall t \geq 0, i \in \mathbb{N}_{1:m}$, then $\bar{\mathcal{C}}_i$ is forward invariant under any $u$ satisfying \eqref{eq:aDRCBF_condition} (See Theorem 4 in \cite{xiao2021high}).
% %
% However, $\Delta_i(t) \geq 0$ for all $t \geq 0$, is generally impossible since $d \neq 0$.
% %
% Next, we will show that $\Gamma_i$ can still render the safe set $\bar{\mathcal{C}}$ forward invariant for the case when $\Delta_i \geq 0$ fails to be satisfied for all $t$.
% We prove this by showing that $\Delta_i \geq 0$ can be strictly satisfied when $x$ approaches $\partial \mathcal{C}_{i}$.
%

Suppose that there exists a time instant $t_1$ such that the state trajectory approaches $\partial\mathcal{C}_{l}$ as $t\rightarrow t_1$ for some $l \in \mathbb{N}_{1:m}$, i.e., {\wxyr $\tilde{\phi}_{l-1}(t) \rightarrow 0$ as $t\rightarrow t_1$ while $\tilde{\phi}_{l-1}(t)>0 $ for all $t\in[0,t_1)$.}
Since {\wxyr the control input} $u$ is locally Lipschitz and $x(0) \in \bar{\mathcal{C}}$,  {\wxyr the closed-loop system admits} a unique continuous solution, {\wxyr and consequently}, the state trajectory must enter $\mathcal{N}(\partial \mathcal{C}_l)\cap\bar{\mathcal{C}}$ before approaching $\partial \mathcal{C}_l$.
{\wxyr Let $t_2 < t_1$ be the smallest time instant such that $x(t) \in \mathcal{N}(\partial \mathcal{C}_l)\cap\bar{\mathcal{C}}$ for all $t \in [t_2, t_1)$, and thus $\tilde{\phi}_{m-1}(x(t_2))>0$.
Lemma \ref{lm:b_i>b_j_adaptive} and \eqref{eq:thm_2_proof:3} imply that}
\begin{align} \label{eq:thm_2_proof:4}
    \dot{\tilde{\phi}}_{i-1}(x) + p_i\tilde{\phi}_{i-1}(x) {\wxyr \geq} \tilde{\phi}_i(x) + {\wxyr \mathcal{H}}, ~i \in \mathbb{N}_{l:m}
\end{align}
holds for all $x \in \mathcal{N}(\partial \mathcal{C}_l) \cap \bar{\mathcal{C}}$.
From \eqref{eq:aDRCBF_condition} and \eqref{eq:thm_2_proof:4}, one has $$\dot{\tilde{\phi}}_{m-1}(x(t)) + p_m\tilde{\phi}_{m-1}(x(t)) {\wxyr \geq \mathcal{H}},~~t\in[t_2,t_1).$$
%
% Now we prove $\tilde{\phi}_{m-1}(x(t)) > 0$ for all $t \in [t_2,t_1]$.
% %
% Construct the following auxiliary system $$\dot{y}(t) = -p_my(t) + \mathcal{H},~~~y(t_2) = \tilde{\phi}_{m-1}(x(t_2)).$$
% %
% Using Comparison Lemma (Lemma 3.4, \cite{khalil2002nonlinear}), we have
% $$\tilde{\phi}_{m-1}(x(t)) \geq e^{-p_m (t-t_2)} \tilde{\phi}_{m-1}(x(t_2)) + \left(1- e^{-p_m(t-t_2)}\right)\frac{\mathcal{H}}{p_m}$$
% for all $t \in [t_2,t_1^-]$.
%
Consider the auxiliary system $\dot{y}(t) = -p_my(t) + {\wxyr \mathcal{H}}$ {\wxyr with $y(t_2) = \tilde{\phi}_{m-1}(x(t_2))$.}
Using Comparison Lemma (Lemma 3.4, \cite{khalil2002nonlinear}), {\wxyr one has $\tilde{\phi}_{m-1}(x(t))\geq y(t) > 0$ and further $\dot{\tilde{\phi}}_{m-2}(x(t)) + p_{m-1}\tilde{\phi}_{m-2}(x(t)) \geq \mathcal{H},~\forall t\in[t_2,t_1)$ by \eqref{eq:thm_2_proof:4}.}
%
% Similarly, using Comparison Lemma, one has $\tilde{\phi}_{m-2}(x(t))> 0$ for all $t \in [t_2,t_1)$.}
%
By repeating these steps, {\wxyr one can show that $\tilde{\phi}_{l}(x(t))> 0$ and further $\dot{\tilde{\phi}}_{l-1}(x(t)) + p_l\tilde{\phi}_{l-1}(x(t)) \geq \mathcal{H}$ for all $t \in [t_2,t_1)$ according to \eqref{eq:thm_2_proof:4}.
Denote $t_1^-$ as the time instant just before $t_1$.
By continuity, $\tilde{\phi}_{l-1}(x(t_1^-))$ can be arbitrarily small such that $\mathcal{H}> p_l\tilde{\phi}_{l-1}(x(t_1^-))$, which implies $\dot{\tilde{\phi}}_{l-1}(x(t_1^-)) > 0$. However, since $t_1$ is assumed to be the first time instant when $\tilde{\phi}_{l-1}$ reaches zero, it must hold that $\dot{\tilde{\phi}}_{l-1}(x(t_1^-)) \leq 0$, which contradicts the above analysis.}
Since the non-existence of any trajectory approaching $\partial \mathcal{C}_l$ can be shown for arbitrary $l\in\mathbb{N}_{1:m}$, one can conclude that the set $\bar{\mathcal{C}}$ is forward invariant for \eqref{eq:dis_sys}, which further implies $b(x(t))> 0, ~\forall t\geq 0$. 
\end{proof}

\subsection{Optimal Control with {\wxyr rDRCBF}}
By virtue of Theorem \ref{thm:t2}, any Lipschitz continuous controller $u$ satisfying \eqref{eq:aDRCBF_condition} can {\wxyr render the set $\bar{\mathcal{C}}$ forward invariant for system \eqref{eq:dis_sys}.}
Now, we integrate ISS-CLF with {\wxyr rDRCBF} to reformulate Problem \ref{problem:p1} as the following {\wxyr quadratic program}
\begin{align*}
    &\operatorname*{min}_{(u,{\wxyr \tilde\delta}) \in \mathbb{R}^p \times \mathbb{R}} 
   ~~ J(x, u)+ \rho \tilde{\delta}^2~~~~~~~~~~~~~\\
    &~~~~\text{s.t. }~~~ L_fV(x)+L_gV(x)u \leq -\sigma V(x)+\tilde{\delta}\\
    &~~~~~~~~~~~~ \tilde{\pi}_m(x) + \tilde{\beta}_u(x) u +\sum_{j=0}^{m-1} c_j^m \tilde{\psi}_j(x) \geq k_m \Gamma_{m-1}(x),
\end{align*}
where $\tilde{\delta}$ is a slack variable {\wxyr to relax ISS-CLF condition} for guaranteeing {\wxyr the quadratic program to be} feasible.

\begin{remark} \label{rm:para_design_adrcbf}
    (\textbf{Parameter design of {\wxyr rDRCBFs}})
    In {\wxyr rDRCBF}, $k_i$ and $r_i$ are introduced.
    The parameter $k_i$ is used as shown in Remark \ref{rm:para_design_drcbf} to balance the conservativeness brought by over-estimation of $\frac{1}{4k_i}\lVert L_h \tilde{\psi}_{i-1}\rVert^2$ and $k_i \mathcal{D}^2$.
    %
    % It can be obviously observed that a large $k_i$ simultaneously enhances $\frac{1}{4k_i}\lVert L_h \tilde{\psi}_{i-1}\rVert^2$ and suppresses $k_i\lVert D\rVert^2$.
    To reduce conservativeness, $k_i$ can be chosen as $k_i^*$ (See Remark \ref{rm:para_design_drcbf}).
    The parameter $r_i$ is used to further suppress the conservativeness introduced by $\Gamma_{i-1}$.
    {\wxyr By \eqref{eq:B_property}, $B(\tilde{\phi}_{i-1})$} approaches infinity as state trajectories approach $\partial \mathcal{C}_i$, and thereby {\wxyr a smaller $r_i$ makes $\Gamma_{i-1} \geq \mathcal{D}^2$ in a smaller neighborhood of $\partial\mathcal{C}_i$.}
    In other words, $x$ is allowed to get closer to the boundary of the safe set.
\end{remark}

{\wxyr
\begin{remark}
    (\textbf{Robustness of HOCBF, kDRCBF and rDRCBF}) By using extended class $\mathcal{K}_\infty$ functions, HOCBF is well defined in $\mathbb{R}^n$ for system \eqref{eq:dis_sys} if $d$ is sufficiently differentiable, and its robustness is characterized by asymptotic stability of a subset of $\mathscr{C}$, i.e., $\mathscr{X}$.
    As shown in \cite{tan2022high}, this property of $\mathscr{X}$ guarantees that the state $x\in \mathbb{R}^n\setminus\mathscr{X}$ converges asymptotically to $\mathscr{X}$; however, when $d$ is non-differentiable, the higher-order derivatives required by HOCBF are not well defined, and consequently HOCBF may not work in $\mathscr{X}$.
    In contrast, kDRCBF remains well defined in $\mathbb{R}^n$ under general bounded disturbance, and its robustness can be characterized in two aspects: $i$) kDRCBF uses disturbance bound $\mathcal{D}$ to guarantee a subset of $\mathscr{C}$, i.e., $\bar{\mathscr{C}}$ forward invariant under disturbance; $ii$) the set $\bar{\mathscr{C}}$ is asymptotic stable, which can be shown following the similar development as in Proposition 3, \cite{tan2022high}.
    Without using any knowledge of $d$, rDRCBF can also enforce a subset of $\mathscr{C}$, i.e., $\bar{\mathcal{C}}$ forward invariant by using reciprocal-like functions \eqref{eq:B_property} to over-approximate $d$ near the set boundary.
    However, similar to reciprocal CBF which is not defined at the boundary of safe set \cite{ames2017control}, rDRCBF introduces singularities at the boundary of the safe set, and thus it is only well defined within $\bar{\mathcal{C}}$. 
\end{remark}
}

% \begin{remark}
%     (\textbf{Inter-sampling of aDRCBFs}) While we rigorously prove the safety of continuous system under our aDRCBFs, inter-sampling issue, i.e., QP is solved in a discrete-time manner, may cause unsafe behavior of the system during the implementation.
%     %
%     This issue is more severe in aDRCBFs since we allow the state get close enough to the boundary of safe set, where the safety under aDRCBF may be easily violated if intersample time intervals is large and the sampled state $\hat{x}$ cannot make $\Gamma_i(\hat{x})$ larger than $\mathcal{D}^2$.
%     %
%     One feasible solution is to modify the adaptive function as 
%     $$\Gamma_i(x) = r_i B(\tilde{\phi}_i(x) - \varepsilon_i) + \xi_i,$$
%     where $\varepsilon_i$ and $\xi_i$ are two positive parameters introduced to make $\Gamma_i$ grow faster, thereby handling the inter-sampling issue.
% \end{remark}

\section{Simulation Examples}
{\wxyr This section provides applications of the proposed DRCBFs on ACC and unmanned aerial vehicle (UAV) systems.}

\textbf{\wxyr{Example 1:}} Consider the ACC problem in Section \ref{Sec:ACC}.
For \eqref{eq:ACC}, define its performance index as $J = u^\top Hu + Fu$, where $H$ and $F$ are two performance parameters.  
The follower is required to track a desired speed of $v_d = 35\textnormal{m/s}$, i.e., $\lim_{t\rightarrow \infty }v_f(t) = v_d$.
For speed tracking, we consider the ISS-CLF as $V = (v_f -v_d)^2$.
The corresponding ISS-CLF condition is given as $\frac{2}{M}(v_f - v_d)(u - F_r(v_f)) \leq - \sigma (v_f - v_d)^2.$
The simulation parameters are $M = 1650 \textnormal{kg}$, $v_l = 20\textnormal{m/s}$, $f_0 = 0.1 \textnormal{N}, f_1 = 5 \textnormal{N}\cdot\textnormal{s/m}, f_2 = 0.25 \textnormal{N}\cdot\textnormal{s/m}^2$, $D_{min}= 10\textnormal{m}$, $v_f(0) = 13.89 \textnormal{m/s}$ and $D(0) = 100\textnormal{m}$.
Let $k_1 = k_2 =0.1, p_1 = 5, p_2 = 10$ for both DRCBFs, $r_0 = r_1 = 1$ for rDRCBF, and $\sigma = 10$ for ISS-CLF.
% {\wxyr For both DRCBFs,} $k_1 = k_2 =0.1, r_0 = r_1 = 1, p_1 = 5$ and $p_2 = 10$.
%
% These lead to $c_2^2 = 1, c_1^2 = 15, c_0^2 = 50, c_1^1 = 1$ and $c_0^1 =5$.
%
{\wxyr For} optimization, $H = \frac{2}{M^2}$, $F = -\frac{2F_r}{M^2}$ and $\rho = 2$.
In the following, three cases are studied to  illustrate the performance of our proposed {\wxyr kDRCBF} and {\wxyr rDRCBF}. 
%Please be noted that the bound of disturbances is not used in {\wxyr rDRCBF}.

\begin{figure}[htpb]
    \centering
    \begin{subfigure}{0.43\textwidth}
        \centering
        \includegraphics[width=1\textwidth, trim=0 290 0 290, clip]{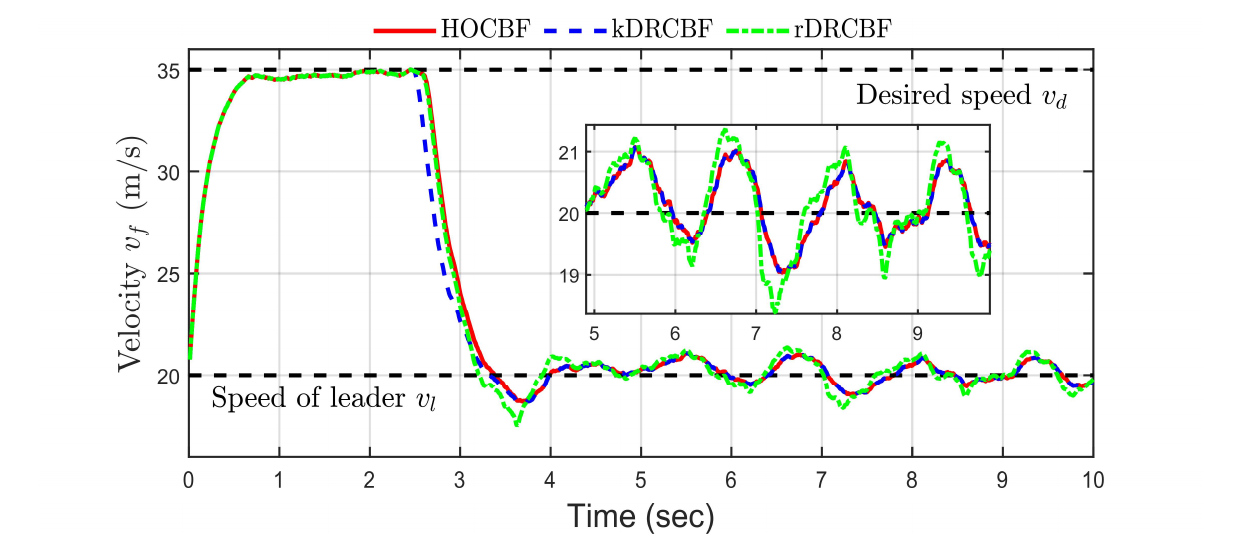}
        \caption{}
        \label{fig:case1_a}
    \end{subfigure}  
    \quad
    \begin{subfigure}{0.43\textwidth}
        \centering
        \includegraphics[width=1\textwidth, trim=0 290 0 290, clip]{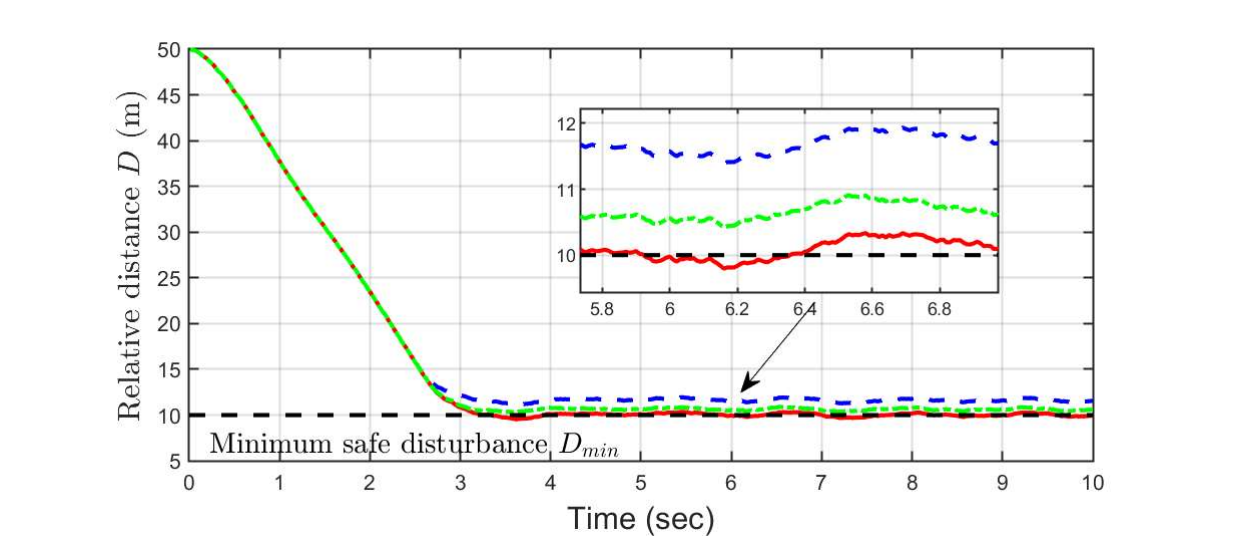}
        \caption{}
        \label{fig:case1_b}
    \end{subfigure}  
    \caption{Profiles of velocity $v_f$ and relative distance $D$ of the ACC system under HOCBF, {\wxyr kDRCBF} and {\wxyr rDRCBF}.}
    \label{fig:case1_combined}
\end{figure}

\noindent{\color{deepblue}\textit{\textbf{Case 1. Safety under non-differentiable disturbances}}}

We compare HOCBF \cite{xiao2021high}, {\wxyr kDRCBF} and {\wxyr rDRCBF} under non-differentiable disturbances.
The disturbances are  $d_u(t) = -4+8\omega_1 +\sin(5t),d_m(t) =  -4 + 8\omega_2 +0.5\cos(10t)$,
% \begin{align*}
%     % d(t) = \begin{pmatrix} -4+\sin(5t) \\ -4+0.5\cos(10t) \end{pmatrix} + 8\text{rand}(2,1),
%     &d_1(t) = -4+8\omega_1 +\sin(5t),\\
%         &d_2(t) =  -4 + 8\omega_2 +0.5\cos(10t).
% \end{align*}
where $\omega_1$ and $\omega_2$ are {\wxyr bounded random signals.}
The results are presented in Fig. \ref{fig:case1_combined}.
Figure \ref{fig:case1_a} shows that under all three CBF approaches, the speed of the following vehicle initially reaches the desired speed, and subsequently slows down {\wxyr when the slack variable works.}
However, as shown by the red line in Fig. \ref{fig:case1_b}, HOCBF fails to maintain the safety distance.
In contrast, both {\wxyr kDRCBF} and {\wxyr rDRCBF} successfully guarantee  $D \geq D_{min}$ for all $t\geq0$. 

\begin{figure}[htpb]
    \centering
    \begin{subfigure}{0.43\textwidth}
        \centering
        \includegraphics[width=1\textwidth, trim=0 290 0 290, clip]{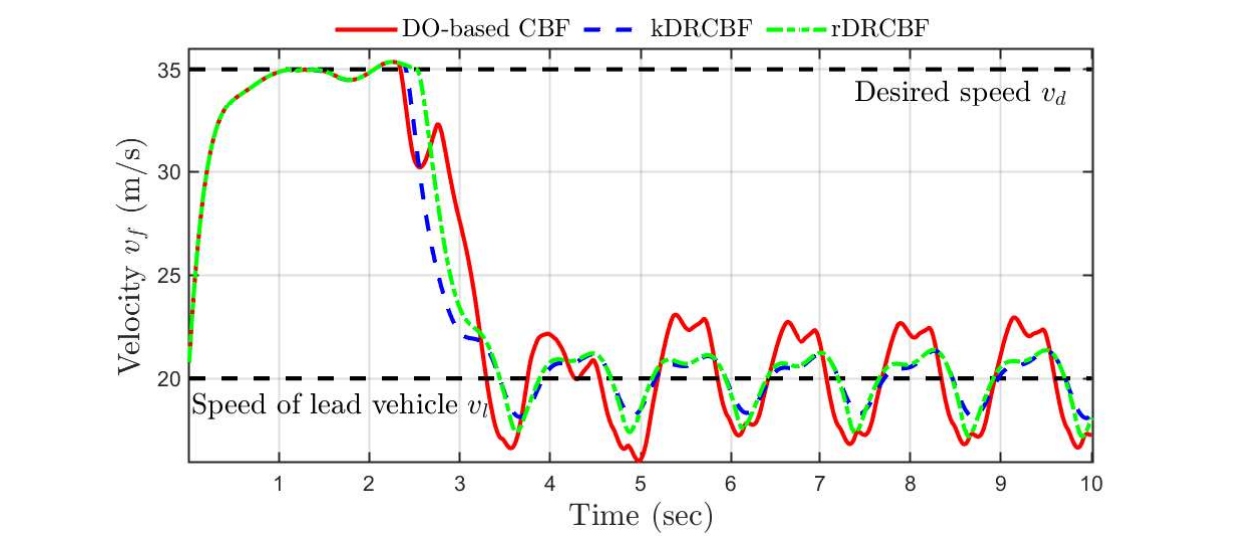}
        \caption{}
        \label{fig:case2_a}
    \end{subfigure}  
    \quad
    \begin{subfigure}{0.43\textwidth}
        \centering
        \includegraphics[width=1\textwidth, trim=0 290 0 290, clip]{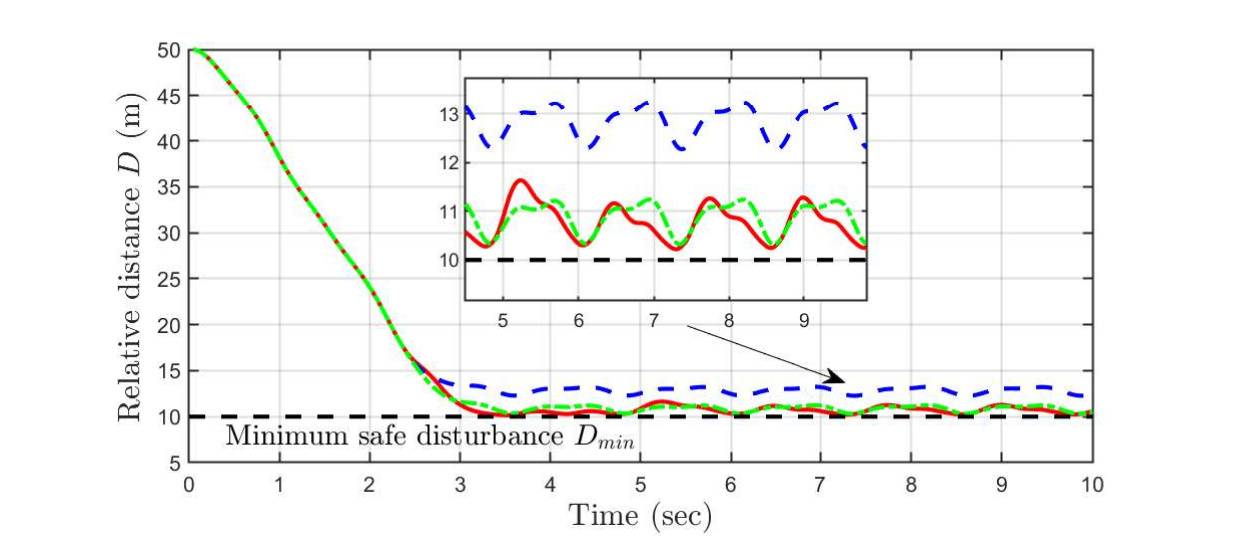}
        \caption{}
        \label{fig:case2_b}
        \end{subfigure}  
    \caption{Profiles of velocity $v_f$ and relative distance $D$ of the ACC system under DO-based CBF, {\wxyr kDRCBF} and {\wxyr rDRCBF}.}
    \label{fig:case2_combined}
\end{figure}

\noindent{\color{deepblue}{\textbf{\textit{Case 2. Safety under differentiable disturbances}}}}
 
This case compares the performance of our proposed {\wxyr kDRCBF} and {\wxyr rDRCBF} approaches with DO-based CBF \cite{wang2025safety}.   
% A bounded continuous disturbance is considered to compare the disturbance rejection performance of the DO-based CBF \cite{wang2025safety}, the proposed DRCBF, and the proposed aDRCBF.
% %
% The DO-based CBF \cite{wang2025safety} represents the state-of-the-art robust CBF for unmatched disturbance rejection.
%
The disturbances are $d_u(t) = 2\sin(5t) + 1.5\cos(10t),~d_m(t)=\sin(10t) + 2\cos(6t)$.
The results are shown in Fig. \ref{fig:case2_combined}.
%
% As seen in Fig. \ref{fig:case2_a}, the velocity trajectories under the DO-based CBF, DRCBF and aDRCBF exhibit only slight differences before $t = 10\textnormal{s}$.
% %
% However, the velocity trajectory under the DO-based CBF shows larger oscillations than those under DRCBF and aDRCBF because DO cannot provide an precise estimation of the high-frequency disturbance.
%
Figure \ref{fig:case2_b} shows that all three CBF approaches can guarantee safety under differentiable disturbances. 
The DO-based CBF achieves safety in a less conservative manner compared to {\wxyr kDRCBF} because our {\wxyr kDRCBF} accounts for the worst-case disturbance, whereas {\wxyr DO-based CBF uses bounds of disturbance derivatives to estimate and compensate disturbances.}
Our {\wxyr rDRCBF} achieves performance comparable to that of the DO-based CBF, without knowing any information of the disturbances.

\begin{figure}[htpb]
    \centering
        \includegraphics[width=0.43\textwidth, trim=0 290 0 300, clip]{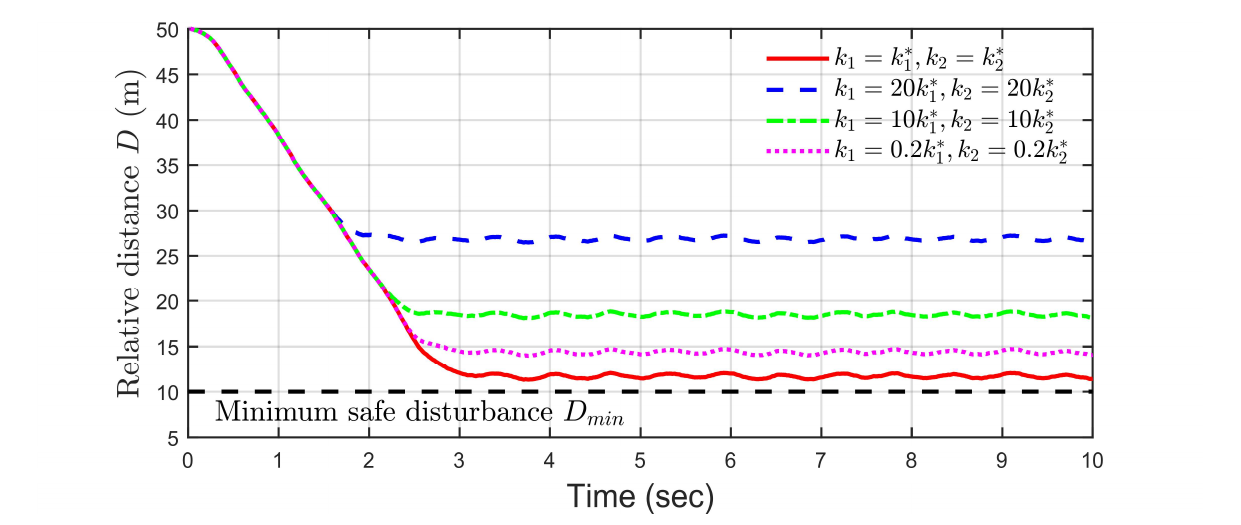}
        \caption{Performance of {\wxyr kDRCBF} with different  $k_1$ and $k_2$}
        \label{fig:case3_a}
    \label{fig:case3_combined}
\end{figure}

\begin{figure}[htpb]
        \centering
        \includegraphics[width=0.43\textwidth, trim=0 250 0 310, clip]{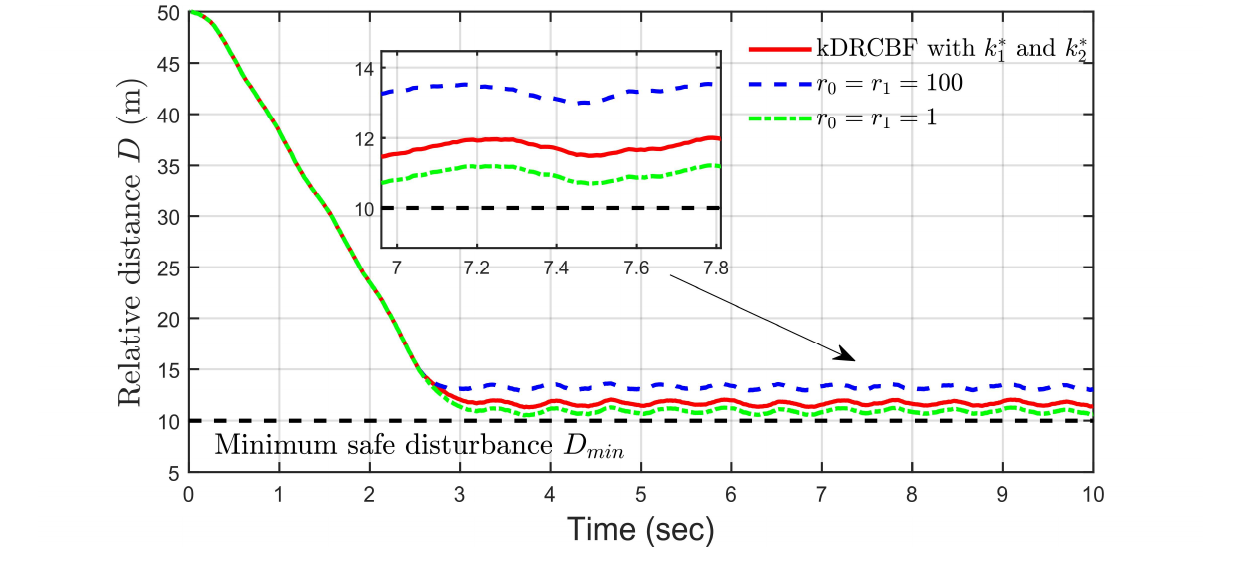}
        \caption{Performance of {\wxyr rDRCBF} with different $r_0$ and $r_1$}
        \label{fig:case3_b}
\end{figure}

% \subsubsection{Case 3: Reducing conservativeness in DRCBF}
\noindent{\color{deepblue}\textit{\textbf{Case 3. Reducing conservativeness}}}
 
This case shows how the parameters $k_i$ and $r_i$ {\wxyr affect} conservativeness in {\wxyr kDRCBF} and {\wxyr rDRCBF}.
%
% {\wxyr Let disturbances satisfy $\lvert d_u\rvert \leq \mathcal{D}_1$ and $\lvert d_m\rvert \leq \mathcal{D}_2$ for two positive $\mathcal{D}_1,\mathcal{D}_2$.}
% \begin{align*}
%         % &d_1(t) = -4 + 8\text{rand}(1,1) +5\sin(2t),\\
%         % &d_2(t) =  -5 + 10\text{rand}(1,1) +4\sin(2t).
%         &d_1(t) = -4 + 8\omega_3 +5\sin(2t),\\
%         &d_2(t) =  -5 + 10\omega_4 +4\sin(2t).
% \end{align*}
%
% where $\omega_3$ and $\omega_4$ are bounded random signals.
%
From Remark \ref{rm:para_design_drcbf}, the optimal parameters for {\wxyr kDRCBF} are $k_1^* = \frac{1}{2\mathcal{D}_1}$ and $k_2^*= \frac{1}{2\mathcal{D}_2}$, {\wxyr where $\mathcal{D}_1$ and $\mathcal{D}_2$ are the bounds of $d_u$ and $d_m$, respectively}.
Simulations are conducted for {\wxyr kDRCBF} when $(k_1,k_2)$ takes different values of $(k_1^*, k_2^*), ~(20k_1^*, 20k_2^*), ~(10k_1^*, 10k_2^*)$ and $(0.2k_1^*, 0.2k_2^*)$.
Figure \ref{fig:case3_a} shows that $k_1^*$ and $k_2^*$ are the optimal choices for reducing conservativeness.
%
% With the optimal parameters $k_1^*$ and $k_2^*$, we let $k_1 = k_1^*$ and $k_2 = k_2^*$.
%
% Next, we run the ACC system using DRCBF with $k_1 = k_1^*,k_2=k_2^*$ and 
For {\wxyr rDRCBF}, we use two sets of parameters $r_0 = r_1=100$ and $r_0 = r_1=1$. 
Figure \ref{fig:case3_b} illustrates that large values of $r_0$ and $r_1$ may make {\wxyr rDRCBF} performance even more conservative than that of {\wxyr kDRCBF} with $k_1^*$ and $k_2^*$.
The conservativeness is reduced and outperforms {\wxyr kDRCBF} with $k_1^*$ and $k_2^*$ as $r_0, r_1$ decrease to $1$.
This observation aligns with the analysis in Remark \ref{rm:para_design_adrcbf}.
In conclusion, optimal parameters $k_1^*$ and $k_2^*$ achieve the least conservativeness for {\wxyr kDRCBF}, and small gains $r_0$ and $r_1$ can further reduce conservativeness for {\wxyr rDRCBF}. 

{\wxyr\textbf{Example 2:} Consider the following UAV system
% $\dot{p}_x = v_x,\dot{p}_y= v_y,\dot{p}_z = v_z,\dot{v}_x =  -0.1\cos(0.08p_x) + u_x + d_x,\dot{v}_y =  0.2\sin(0.05p_y)\cos(0.1p_z) + u_y + d_y,\dot{v}_z =  0.15\cos(0.08p_x)\cos(0.06p_z) + u_z + d_z$,
\begin{align*}
    &\dot{p}_x = v_x,~~\dot{p}_y= v_y,~~\dot{p}_z = v_z \\
    &\dot{v}_x =  -0.1\cos(0.08p_x) + u_x + d_x,\\
    &\dot{v}_y =  0.2\sin(0.05p_y)\cos(0.1p_z) + u_y + d_y,
    \\
    &\dot{v}_z =  0.15\cos(0.08p_x)\cos(0.06p_z) + u_z + d_z,
\end{align*}
where $p_x,p_y,p_z$, $v_x,v_y,v_z$, $u_x,u_y,u_z$ and $d_x,d_y,d_z$ are the position, velocity, input and disturbance of UAV along $X$, $Y$ and $Z$ axis.
The UAV aims to approach a target while avoiding collision with a moving obstacle $\dot{p}_x^o = v_x^o,~\dot{p}_y^o = v_y^o,~\dot{p}_z^o = v_z^o$, where $p_o,~p_o,~p_o$ and $v_x^o,~v_y^o,~v_z^o$ are the position and velocity of obstacle along $X$, $Y$ and $Z$ axis.
The safety constraint is $b = (p_x-p_x^o)^2 + (p_y-p_y^o)^2+(p_z-p_z^o)^2 -  \gamma^2\geq 0$, where $\gamma$ is the safe distance.
\begin{figure}[htpb]
        \centering
        \includegraphics[width=0.43\textwidth, trim=0 290 0 290, clip]{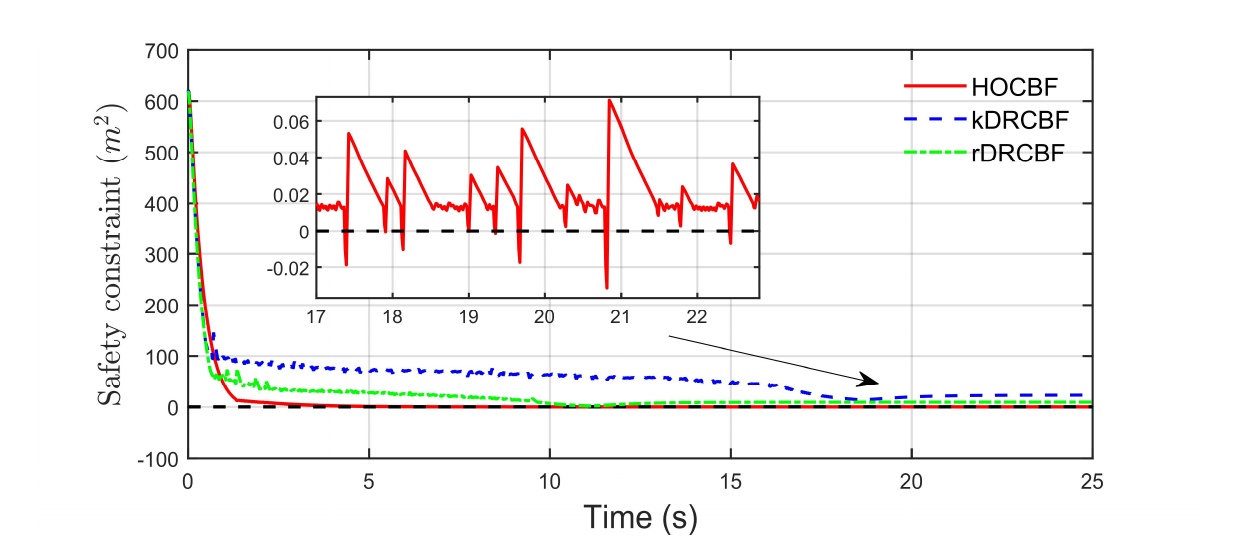}
        \caption{{\wxyr Safety constraint of UAV under HOCBF, kDRCBF and rDRCBF. $b<0$ means that safety constraint is violated.}}
        \label{fig:uav_safety}
\end{figure}
The position of obstacle is known for CBF design, and the velocity of obstacle and disturbances are unknown but bounded.
% $\lvert d_x\rvert, \lvert d_y\rvert, \lvert d_z\rvert,\lvert v_x^o\rvert, \lvert v_y^o\rvert, \lvert v_z^o\rvert$ are bounded (the velocity of obstacle is regarded as an unmatched disturbance).
%
Let the UAV be implemented by HOCBF \cite{xiao2021high}, kDRCBF and rDRCBF, respectively.
As shown in Fig. \ref{fig:uav_safety}, HOCBF fails to achieve collision avoidance, while both kDRCBF and rDRCBF can strictly guarantee safety. }
%

%==================================
\section{Conclusion}
This paper proposed {\wxyr two} disturbance rejection CBFs for guaranteeing safety in the presence of general disturbances.
To reject non-differentiable unmatched disturbances, the disturbance bound is used in {\wxyr kDRCBF} to prevent {\wxyr involving the} derivative of disturbance of any order.
%
% To facilitate its practical application, 
{\wxyr To avoid using any knowledge of disturbances, an rDRCBF was further proposed} by introducing {\wxyr a reciprocal-like} term to overly approximate the disturbances.
{\wxyr Moreover, this reciprocal-like} term provides a choice to adjust the degree of conservativeness.

\section*{References}
% Generated by IEEEtran.bst, version: 1.14 (2015/08/26)

% \vskip 1cm
% %%------------------------------------

% The second paragraph uses the pronoun of the person (he or she) and
% not the author’s last name. It lists military and work experience, including
% summer and fellowship jobs. Job titles are capitalized. The current job must
% have a location; previous positions may be listed without one. Information
% concerning previous publications may be included. Try not to list more than
% three books or published articles. The format for listing publishers of a book
% within the biography is: title of book (publisher name, year) similar to a
% reference. Current and previous research interests end the paragraph.

% The third paragraph begins with the author’s title and last name (e.g.,
% Dr. Smith, Prof. Jones, Mr. Kajor, Ms. Hunter). List any memberships in
% professional societies other than the IEEE. Finally, list any awards and work
% for IEEE committees and publications. If a photograph is provided, it should
% be of good quality, and professional-looking.
% \end{IEEEbiography}

% \begin{IEEEbiographynophoto}{Third C. Author Jr.} (Member, IEEE), photograph and biography not available at the
% time of publication.
% \end{IEEEbiographynophoto}

\end{document}